\documentclass[a4paper,UKenglish,cleveref]{lipics-v2021}

\hideLIPIcs
\nolinenumbers

\usepackage[ruled]{algorithm}
\newcommand{\psfrage}[1]{{\color{blue}{\sf[PS: #1]}}} %
\newcommand{\hpfrage}[1]{{\color{violet}\sf[HP: #1]}} %
\newcommand{\gvfrage}[1]{{\color{teal}\sf[GV: #1]}} %
\newcommand{\pffrage}[1]{{\color{red}\sf[PF: #1]}} %
 \renewcommand{\psfrage}[1]{} \renewcommand{\hpfrage}[1]{} \renewcommand{\gvfrage}[1]{} \renewcommand{\pffrage}[1]{}

\usepackage{xcolor}
\usepackage{xspace}
\usepackage{booktabs}
\usepackage{dsfont}
\usepackage{marvosym}
\usepackage{placeins}
\usepackage{rotating}
\usepackage[nocompress]{cite}

\usepackage{pgfplots}

\usepgfplotslibrary{groupplots}
\usetikzlibrary{pgfplots.statistics}
\usetikzlibrary{backgrounds}

\pgfdeclareplotmark{flippedTriangle}{%
  \pgfpathmoveto{\pgfqpointpolar{-90}{1.2\pgfplotmarksize}}%
  \pgfpathlineto{\pgfqpointpolar{30}{1.2\pgfplotmarksize}}%
  \pgfpathlineto{\pgfqpointpolar{150}{1.2\pgfplotmarksize}}%
  \pgfpathclose%
  \pgfusepath{stroke}%
}
\pgfdeclareplotmark{rightTriangle}{%
  \pgfpathmoveto{\pgfqpointpolar{0}{1.2\pgfplotmarksize}}%
  \pgfpathlineto{\pgfqpointpolar{120}{1.2\pgfplotmarksize}}%
  \pgfpathlineto{\pgfqpointpolar{240}{1.2\pgfplotmarksize}}%
  \pgfpathclose%
  \pgfusepath{stroke}%
}
\pgfdeclareplotmark{leftTriangle}{%
  \pgfpathmoveto{\pgfqpointpolar{-180}{1.2\pgfplotmarksize}}%
  \pgfpathlineto{\pgfqpointpolar{-60}{1.2\pgfplotmarksize}}%
  \pgfpathlineto{\pgfqpointpolar{60}{1.2\pgfplotmarksize}}%
  \pgfpathclose%
  \pgfusepath{stroke}%
}
\pgfdeclareimage[interpolate,height=1.85mm,width=1.85mm]{lemonMark}{fig/lemon_yellow}
\pgfdeclareplotmark{lemon}{\pgftext[at=\pgfpointorigin]{\pgfuseimage{lemonMark}}}
\pgfdeclareimage[interpolate,height=1.85mm,width=1.85mm]{lemonMarkVl}{fig/lemon_yellow_vl}
\pgfdeclareplotmark{lemonVl}{\pgftext[at=\pgfpointorigin]{\pgfuseimage{lemonMarkVl}}}

\definecolor{colorDirectRankStoring}{HTML}{F8BA01}
\definecolor{veryLightGrey}{HTML}{F2F2F2}
\definecolor{colorHollowTrie}{HTML}{000000}
\definecolor{colorHollowTrieDist}{HTML}{000000}
\definecolor{colorCentroidHollowTrie}{HTML}{000000}
\definecolor{colorPaCoTrieJava}{HTML}{4DAF4A}
\definecolor{colorPathDecomposedTrie}{HTML}{984EA3}
\definecolor{colorVllcp}{HTML}{E41A1C}
\definecolor{colorLcp}{HTML}{E41A1C}
\definecolor{colorTwoStepsLcp}{HTML}{E41A1C}
\definecolor{colorZFastTrieDistributor}{HTML}{377EB8}
\definecolor{color1}{HTML}{444444}
\definecolor{color2}{HTML}{A65628}
\definecolor{color3}{HTML}{377EB8}

\pgfplotscreateplotcyclelist{myColorList}{%
  colorDirectRankStoring,mark=flippedTriangle\\%
  colorHollowTrie,mark=triangle\\%
  colorPaCoTrieJava,mark=pentagon\\%
  colorPathDecomposedTrie,mark=square\\%
  colorVllcp,mark=x\\%
  colorZFastTrieDistributor,mark=diamond\\%
}

\pgfplotsset{
  compat=newest,
  mark repeat*/.style={
    scatter,
    scatter src=x,
    scatter/@pre marker code/.code={
      \pgfmathtruncatemacro\usemark{
        or(mod(\coordindex,#1)==0, (\coordindex==(\numcoords-1))
      }
      \ifnum\usemark=0
        \pgfplotsset{mark=none}
      \fi
    },
    scatter/@post marker code/.code={}
  },
  major grid style={thin,dotted},
  minor grid style={thin,dotted},
  ymajorgrids,
  yminorgrids,
  every axis/.append style={
    scale only axis,
    line width=0.7pt,
    tick style={
      line cap=round,
      thin,
      major tick length=4pt,
      minor tick length=2pt,
    },
    mark options={solid},
  },
  legend cell align=left,
  legend style={
    line width=0.7pt,
    /tikz/every even column/.append style={column sep=3mm,black},
    /tikz/every odd column/.append style={black},
    mark options={solid},
  },
  legend style={font=\small},
  title style={yshift=-2pt},
  enlarge x limits=0.04,
  every tick label/.append style={font=\footnotesize},
  every axis label/.append style={font=\small},
  every axis y label/.append style={yshift=-1ex},
  /pgf/number format/1000 sep={},
  axis lines*=left,
  xlabel near ticks,
  ylabel near ticks,
  axis lines*=left,
  label style={font=\footnotesize},
  tick label style={font=\footnotesize},
  enlargelimits=0.05,
  plotCompetitorConstruction/.style={
    title style={yshift=-15pt},
    xlabel={Bits/key},
    extra x ticks={13},
    extra x tick labels={$\geq$},
    ylabel={Constr. MKeys/s},
    xlabel shift=-5pt,
    width=42mm,
    height=25mm,
    only marks,
    ymax=9.5,
    xmin=2.1,
    xmax=13.5,
  },
  plotCompetitorQueries/.style={
    title style={yshift=-15pt},
    xlabel={Bits/key},
    ylabel={Query MKeys/s},
    extra x ticks={13},
    extra x tick labels={$\geq$},
    xlabel shift=-5pt,
    width=42mm,
    height=25mm,
    only marks,
    ymax=2.8,
    xmin=2.1,
    xmax=13.5,
  },
  plotThresholds/.style={
    title style={yshift=-15pt},
    xlabel={Recursion threshold},
    ylabel={Bits/key},
    xtick={64,128,256,512},
    xlabel shift=-2pt,
    width=42mm,
    height=25mm,
    ymax=7.2,
    ymin=5.8,
    xmin=60,
    xmax=540,
    cycle list name=myColorList,
    log ticks with fixed point,
  },
  plotThresholdsQ/.style={
    title style={yshift=-15pt},
    xlabel={},
    ylabel={Query time},
    xtick={64,128,256,512},
    xticklabels={,,,},
    xlabel shift=-5pt,
    width=42mm,
    height=25mm,
    ymax=4000,
    ymin=800,
    xmin=60,
    xmax=540,
    cycle list name=myColorList,
    log ticks with fixed point,
  }
}

\usepackage{pgfplots}
\crefname{listing}{Algorithm}{Algorithms}

\newcommand{\us}{\ensuremath{u}}%
\newcommand{\universe}{\ensuremath{[u]}}
\newcommand{\overhead}{\ensuremath{\eta}}
\newcommand{\eps}{\ensuremath{\varepsilon}}

\newcommand{\myunderline}[1]{{\kern-0.05em\underline{\kern0.05em #1\kern-0.05em}\kern0.05em}}
\newcommand{\myoverline}[1]{{\kern0.05em\overline{\kern-0.05em #1\kern0.05em}\kern-0.05em}}
\newcommand{\etal}{et~al.\@}
\newcommand{\Oh}[1]{\mathcal{O}\!\left( #1\right)}

\newcommand{\lemon}{LeMonHash\xspace}
\newcommand{\lemonvl}{LeMonHash-VL\xspace}

\newcommand{\mytitle}{Learned Monotone Minimal Perfect Hashing}
\title{\mytitle}
\titlerunning{\mytitle}

\author{Paolo Ferragina}{University of Pisa, Italy}{paolo.ferragina@unipi.it}{https://orcid.org/0000-0003-1353-360X}{}
\author{Hans-Peter Lehmann}{Karlsruhe Institute of Technology, Germany}{hans-peter.lehmann@kit.edu}{https://orcid.org/0000-0002-0474-1805}{}
\author{Peter Sanders}{Karlsruhe Institute of Technology, Germany}{sanders@kit.edu}{https://orcid.org/0000-0003-3330-9349}{}
\author{Giorgio Vinciguerra}{University of Pisa, Italy}{giorgio.vinciguerra@unipi.it}{https://orcid.org/0000-0003-0328-7791}{}

\newcommand{\myauthorrunning}{P. Ferragina, H.-P. Lehmann, P. Sanders, G. Vinciguerra}
\authorrunning{\myauthorrunning}
\Copyright{Paolo Ferragina, Hans-Peter Lehmann, Peter Sanders, and Giorgio Vinciguerra}
\hypersetup{
  colorlinks=true,
  pdftitle={\mytitle},
  pdfauthor={\myauthorrunning},
  pdfsubject={}
}


\supplementdetails[subcategory={Source Code}]{Software}{https://github.com/ByteHamster/LeMonHash}

\relatedversiondetails{An extended abstract of this paper appears in the Proceedings of the 31st Annual European Symposium on Algorithms (ESA 2023)} {https://doi.org/10.4230/LIPIcs.ESA.2023.46}

\funding{
\flag{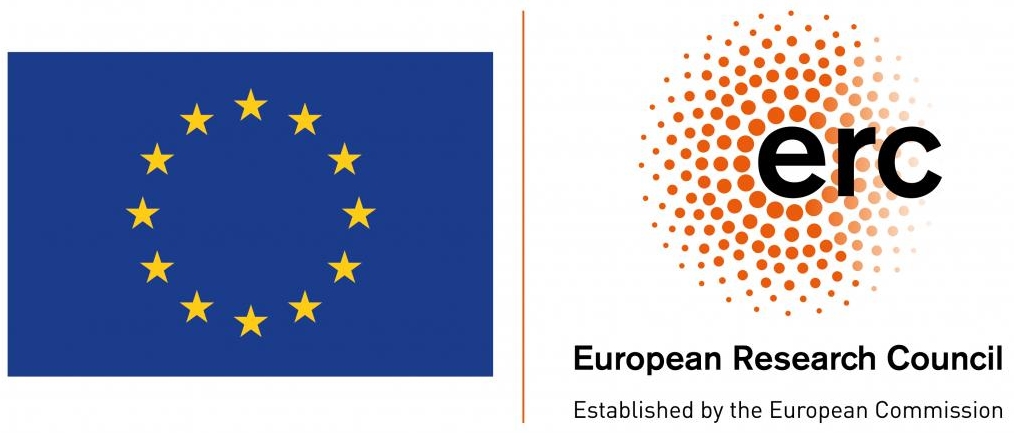}
\looseness=-1
This project has received funding from the European Research Council (ERC) under the European Union’s Horizon 2020 research and innovation programme (grant agreement No. 882500).
PF and GV have been supported by the European Union – Horizon 2020 Program under the scheme ``INFRAIA-01-2018-2019 -- Integrating Activities for Advanced Communities'', Grant Agreement n. 871042, ``SoBigData++: European Integrated Infrastructure for Social Mining and Big Data Analytics'' (http://www.sobigdata.eu), by the NextGenerationEU -- National Recovery and Resilience Plan (Piano Nazionale di Ripresa e Resilienza, PNRR) -- Project: ``SoBigData.it -- Strengthening the Italian RI for Social Mining and Big Data Analytics'' -- Prot. IR0000013 -- Avviso n. 3264 del 28/12/2021, by the spoke ``FutureHPC \& BigData'' of the ICSC -- Centro Nazionale di Ricerca in High-Performance Computing, Big Data and Quantum Computing funded by European Union -- NextGenerationEU -- PNRR, by the  Italian Ministry of University and Research ``Progetti di Ri\-le\-van\-te In\-te\-res\-se Na\-zio\-nale'' project: ``Multicriteria data structures and algorithms'' (grant n. 2017WR7SHH).
}

\acknowledgements{We thank Stefan Walzer for early discussions leading to this paper.}

\EventEditors{Inge Li G{\o}rtz, Martin Farach-Colton, Simon J. Puglisi, and Grzegorz Herman}
\EventNoEds{4}
\EventLongTitle{31st Annual European Symposium on Algorithms (ESA 2023)}
\EventShortTitle{ESA 2023}
\EventAcronym{ESA}
\EventYear{2023}
\EventDate{September 4--6, 2023}
\EventLocation{Amsterdam, the Netherlands}
\EventLogo{}
\SeriesVolume{274}
\ArticleNo{6}

\begin{document}
\maketitle

\def\maxSpeedupNormalDistDRS{16}

\def\maxSpeedupOsmDRS{19}

\def\maxSpeedupUrlRecursiveDRS{7}

\def\maxSpeedupDnaRecursiveDRS{3}

\def\maxSpeedupConstructionNormalDistDRS{2}

\begin{abstract}
  A Monotone Minimal Perfect Hash Function (MMPHF) constructed on a set $S$ of keys is a function that maps each key in $S$ to its rank.
  On keys not in $S$, the function returns an arbitrary value.
  Applications range from databases, search engines, data encryption, to pattern-matching algorithms.

  In this paper, we describe \includegraphics[height=2.3mm]{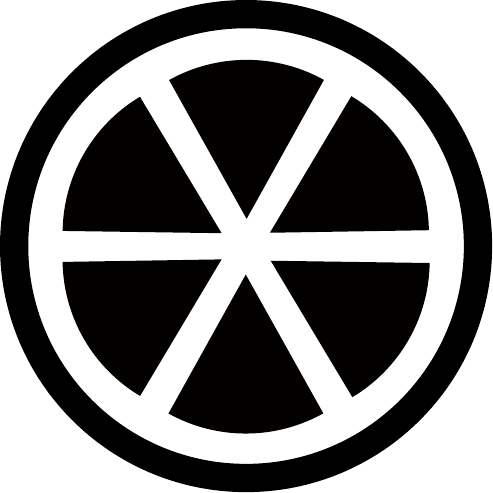}\hspace{.6mm}\lemon, a new technique for constructing MMPHFs for integers.
  The core idea of \lemon is surprisingly simple and effective: we learn a monotone mapping from keys to their rank via an error-bounded piecewise linear model (the PGM-index), and then we solve the collisions that might arise among keys mapping to the same rank estimate by associating small integers with them in a retrieval data structure (BuRR).
  On synthetic random datasets, \lemon needs 34\% less space than the next larger competitor, while achieving about \maxSpeedupNormalDistDRS{} times faster queries.
  On real-world datasets, the space usage is very close to or much better than the best competitors, while achieving up to \maxSpeedupOsmDRS{} times faster queries than the next \emph{larger} competitor.
  As far as the construction of \lemon is concerned, we get an improvement by a factor of up to \maxSpeedupConstructionNormalDistDRS{}, compared to the competitor with the next best space usage.

  We also investigate the case of keys being variable-length strings, introducing the so-called \lemonvl: it needs space within 13\% of the best competitors while achieving up to \maxSpeedupDnaRecursiveDRS{} times faster queries than the next larger competitor.
\end{abstract}

\section{Introduction}
Given a set $S$ of $n$ keys drawn from a universe $\universe =\{0, \dots, \us-1\}$, a \emph{Monotone Minimal Perfect Hash Function} (MMPHF) is a hash function that maps keys from $S$ to their rank, and returns an arbitrary value for keys not in $S$.
As the name suggests, such a function is both \emph{perfect} because it has no collisions on $S$, and \emph{minimal} because its output range is $[n]$.
Differently from a \emph{Minimal Perfect Hash Function} (MPHF)~\cite{fox1992faster,belazzougui2009hash,pibiri2021pthash,muller2014retrieval,chapman2011meraculous,limasset2017fast,lehmann2022sichash,esposito2020recsplit}, which maps keys from $S$ bijectively to $[n]$ in any order, and from an \emph{Order-Preserving MPHF} (OPMPHF)~\cite{Fox:1991}, which retains a given (arbitrary) order on the keys, an MMPHF takes advantage of the natural order of the universe to rank the keys in $S$ in small space, i.e. without encoding them.
Indeed, encoding $S$ needs $\log{\us \choose n}/n = \Omega(\log \tfrac{\us}{n})$ bits per key, and encoding the ranks via an OPMPHF needs $\log(n!)/n = \Omega(\log n)$ bits per key, whilst an MMPHF may use as few as $\Oh{\log \log \log \us}$ bits per key~\cite{belazzougui2009monotone}, which was recently proven to be optimal~\cite{assadi2023tight}.
Throughout this paper, $\log x$ stands for $\log_2 x$, and we use the $w$-bit~word~RAM~model.

MMPHFs have numerous applications \cite{assadi2023tight}.
They enable efficient queries both in encrypted data \cite{boldyreva2011order} and databases \cite{kurpicz2023pachash,lim2011silt}.
Further applications can be found in information retrieval, where MMPHFs can be used to index the lexicon~\cite{Witten:1999} or to compute term frequencies \cite{BelazzouguiNV13,navarro2014spaces}, and in pattern matching \cite{belazzougui2020linear,gagie2020fully,grossi2010optimal}, where MMPHFs are applied mostly to integer sequences representing the occurrences of certain characters in a text.

Despite the widespread use of MMPHFs and recent advancements on their asymptotic bounds~\cite{assadi2023tight}, the practical implementations have not made significant progress in terms of new designs and improved space-time performance since their introduction more than a decade ago~\cite{belazzougui2011theoryPractice}, with only some exceptions targeting query time~\cite{grossi2014decomposition}.
As a matter of fact, the solutions in~\cite{belazzougui2011theoryPractice} are very sophisticated and well-optimised, and they offer a vast number of efficient space-time trade-offs that were hard to beat.

In this paper, we offer a fresh new perspective on MMPHFs that departs from existing approaches, which are mostly based on a trie-like data structure on the keys.
We build upon recent advances in (learning-based) indexing data structures, namely the PGM-index~\cite{ferragina2020pgm,Ferragina:2021tcs}, and in retrieval data structures (or static functions), namely BuRR~\cite{dillinger2022burr}.
The former learns a piecewise linear approximation mapping keys in $S$ to their rank estimate.
The latter allows associating a small fixed-width integer to each key in $S$, without storing~$S$.
We combine these two seemingly unrelated data structures in a surprisingly simple and effective way.
First, we use the PGM to monotonically map keys to buckets according to their rank estimate, and we store the global rank of each bucket's first key in a compressed data structure.
Second, since the rank estimate of some keys might coincide, we solve such bucket collisions by storing the local ranks of these keys using BuRR.
We call our proposal \emph{\lemon}, because it \emph{learns} and \emph{leverages} the smoothness of the input data to build a space-time efficient \emph{monotone} MPHF. On the theoretical side, this achieves $\Oh{1}$ bits per key for inputs which are sufficiently random within buckets\,---\,breaking the superlinear lower bound. Practically, on various integer datasets tried, it needs about one-third less space than previous approaches and is an order of magnitude faster.
We also extend \lemon to support variable-length string keys. This approach needs space within 13\% of the best competitors while being up to 3$\times$ faster.

\subparagraph*{Outline.}
We first describe the basic building blocks of \lemon in \cref{s:prelim} and discuss related work in \cref{s:related}.
In \cref{s:lemon}, we describe \lemon for integers and then extend it to variable-length strings in \cref{s:lemonvl}.
In \cref{s:variants}, we discuss variants and refinements, before proving the space-time guarantees of \lemon in \cref{s:analysis}.
In \cref{s:experiments}, we present our experiments.
In \cref{s:conclusion}, we summarise the paper and give an outlook for future work.

\section{Preliminaries}\label{s:prelim}
In this section, we describe the basic building blocks of \lemon.

\subparagraph*{Bit Vectors.}
Given a bit vector of size $n$ and $b\in\{0,1\}$, the $\textit{rank}_b(x)$ operation returns the number of $b$-bits before position $x$, and the $\textit{select}_b(i)$ operation returns the position of the $i$th $b$-bit.
These operations can be executed in constant time using as little as $o(n)$~bits on top of the bit vector \cite{Jacobson:1989,clark1997compact}, and they have very space-time efficient implementations~\cite{kurpicz2022pasta,gbmp2014sea,vigna2008broadword}.

\subparagraph*{Elias-Fano.}\label{s:eliasFano}
Elias-Fano Coding \cite{Elias74, Fano71} is a way to efficiently store a non-decreasing sequence of $n$ integers over a universe of size $\us$.
An integer at position $i$ is split into two parts.
The $\log n$ upper bits $x$ are stored in a bit vector $H$ as a 1-bit in $H[i + x]$.
The remaining lower bits are directly stored in an array $L$.
Integers can be accessed in constant time by finding the $i$th $1$-bit in $H$ using a $\textit{select}_1$ data structure and by looking up the lower bits in~$L$.
Predecessor queries are possible by determining the range of integers that share the same upper bits of the query key using two $\textit{select}_0$ queries, and then performing a binary search on that range.
If there are no duplicates, this binary search takes  $\Oh{\min \{\log  n, \log \tfrac{\us}{n}\}}$ time.
The space usage of an Elias-Fano coded sequence is $n\lceil \log \tfrac{\us}{n}\rceil + 2n + o(n)$ bits (see \cite[\S4.4]{Navarro:2016book}).
Partitioned Elias-Fano \cite{ottaviano2014partitioned} is an extension that uses dynamic programming to partition the input into multiple independent Elias-Fano sequences to minimise the overall space usage.

\subparagraph*{PGM-index.}\label{s:pgm}
The PGM-index~\cite{ferragina2020pgm} is a space-efficient data structure for predecessor and rank queries on a sorted set of $n$ keys from an integer universe $\universe$.
Given a query $q \in \universe$, it computes a rank estimate that is guaranteed to be close to the correct rank by a given integer parameter~$\eps$.
If one stores the input keys, then the correct rank  can be recovered via an $\Oh{\log \eps}$-time binary search on $2\eps+1$ keys around the rank estimate.
The PGM is constructed in $\Oh{n}$~time by first mapping the sorted integers $x_1, \dots, x_n$ in $S$ to points $(x_1, 1), \dots, (x_n, n)$ in a key-position Cartesian plane, and then learning a piecewise linear $\eps$-approximation of these points, i.e. a sequence of $m$ linear models each approximating the rank of the keys in a certain sub-range of~$\universe$ with a maximum absolute error $\eps$.
The value $m$, which impacts on the space of the PGM, can range between 1 and $m \leq n/(2\eps)$~\cite[Lemma~2]{ferragina2020pgm} depending on the ``approximate linearity'' of the points.
In practice, it is very low and can be proven to be $m = \Oh{n/\eps^2}$ when the gaps between keys are random variables from a proper distribution~\cite{Ferragina:2021tcs}.
The time complexity to compute the rank estimate with a PGM is given by the time to search for the linear model that contains the searched key $q$, which boils down to a predecessor search on $m$ integers from a universe of size $\us$.
For this, there exist many trade-offs in various models of computations~\cite{ferragina2020pgm,Navarro:2020pred}.

\subparagraph*{Retrieval Data Structures.}\label{s:retrieval}
A \emph{retrieval data structure} or \emph{static function} on a set $S$ of $n$ keys denotes a function $f: S\rightarrow\{0,1\}^r$ that returns a specific $r$-bit value for each key.
Applying the function on a key not in $S$ returns an arbitrary value.
Retrieval data structures take $(1+\overhead)rn$~bits, where $\overhead\geq0$ is the \emph{space overhead} over the space lower bound of $rn$ bits.

\emph{MWHC} \cite{majewski1996family} is a retrieval data structure based on hypergraph peeling, has an overhead $\overhead=0.23$ and can be evaluated in constant time.
2-step MWHC \cite{belazzougui2011theoryPractice} can have a smaller overhead than MWHC by using two MWHC functions of different widths. %

The more recently proposed \emph{Bumped Ribbon Retrieval} (BuRR) data structure \cite{dillinger2022burr} basically consists of a matrix. The output value for a key can be obtained by multiplying the hash of the key with that matrix. The matrix can be calculated by solving a linear equation system. Because BuRR uses hash functions with \emph{spacial coupling} \cite{walzer2021peeling}, the equation system is almost a diagonal matrix, which makes it very efficient to solve.
When some rows of the equation system would prevent successful solving, BuRR \emph{bumps} these rows (and the corresponding keys) to the next layer of the same data structure.
BuRR has an overhead $\eta=\Oh{\log W / (rW^2)}$ and can be evaluated in $\Oh{1+rW/\log n}$ time, where $W=\Oh{\log n}$ is a parameter called ribbon width.
In practice, BuRR achieves space overheads well below $\overhead=1\%$ while being faster than widely used data structures with much larger overhead \cite{dillinger2022burr}.

\section{Related Work}\label{s:related}
\label{s:mph}
Non-monotone perfect hash functions are a related and very active area of research \cite{fox1992faster,belazzougui2009hash,pibiri2021pthash,muller2014retrieval,chapman2011meraculous,limasset2017fast,lehmann2022sichash,esposito2020recsplit,bez2022recsplit}.
Due to space constraints, we do not review them in detail.
For a more detailed list, refer to Ref. \cite{lehmann2022sichash}.
We also do not describe order-preserving minimal perfect hash functions~\cite{Fox:1991} because their theoretical lower bound can trivially be reached by using a retrieval data structure taking $\log n$ bits per key (plus a small overhead).
Another loosely related result is using learned models as a replacement for hash functions in traditional hash tables \cite{Kraska:2018,sabek2022can}, but it generally has a negative impact on the probe/insert throughput (and most likely on the space too, due to the storage of the models' parameters, which these studies do not evaluate).
\label{s:mmph}
We now look at monotone minimal perfect hash functions, first describing the idea of bucketing before then continuing with specific MMPHF constructions.

\subparagraph*{Bucketing.}\label{s:bucketing}
Bucketing \cite{belazzougui2011theoryPractice} is a general technique to break down MMPHF construction into smaller sub-problems.
The idea is to store a simple monotone, but not necessarily minimal or perfect \emph{distributor} function that maps input keys to buckets.
Each bucket receives a smaller number of keys that can then be handled using some (smaller) MMPHF data structure.
To determine the global rank of a key, we need the prefix sum of the bucket sizes.
For equally-sized buckets, this is trivial.
Otherwise, this sequence can be stored with Elias-Fano coding.
In the paper by Belazzougui \etal~\cite{belazzougui2011theoryPractice}, where many of the following techniques are described, the authors use MWHC \cite{majewski1996family} to explicitly store the ranks within each bucket.
\lemon uses a learned distributor and buckets of expected size 1 (see \cref{s:lemon}).

\subparagraph*{Longest Common Prefix.}\label{s:lcp}
Bucketing with Longest Common Prefixes (LCP) \cite{belazzougui2009monotone} maps keys to equally sized buckets.
A first retrieval data structure maps all keys to the \emph{length} of the LCP among all keys in its bucket.
A second one then maps the \emph{value} of the LCP to the bucket index.
Overall, it uses $\Oh{\log\log \us}$ bits per key and query time $\Oh{(\log \us)/w}$, and in practice it has been shown to be the fastest but the most space-inefficient MMPHF~\cite{belazzougui2011theoryPractice}.

\subparagraph*{Partial Compacted Trie.}
First map the keys to equally sized buckets and consider the last key of each bucket as a {\em router} indexed by a \emph{compacted trie}, e.g., a binary tree where every node contains a bit string denoting the common prefix of its descending keys.
During queries, the trie is traversed by comparing the bit string of the traversed nodes with the key to decide whether to stop the search operation at some node (if the prefix does not match), or descend into the left or right subtree based on the next bit of the key.
A \emph{Partial Compacted Trie} (PaCo Trie) \cite{belazzougui2011theoryPractice} compresses the compacted trie above by 30--50\% by exploiting the fact that, in an MMPHF, the trie needs to correctly rank only the keys from the input set.
Therefore, each node can store a shorter bit string just long enough to correctly route all input keys.

\subparagraph*{Hollow Trie.}
A \emph{Hollow Trie} \cite{belazzougui2011theoryPractice} only stores the \emph{position} of the next bit to look at.
Hollow tries can be represented succinctly using balanced parentheses \cite{munro2001succinct}.
To use hollow tries for bucketing, and thus allow the routing of not-indexed keys, we need a modification to the data structure.
The \emph{Hollow Trie Distributor} \cite{belazzougui2011theoryPractice} uses a retrieval data structure that maps the compacted substrings of each key in each tree node to the behaviour of that key in the node (stopping at the left or right of the node, or following the trie using the next bit of the key).
Overall, it uses $\Oh{\log\log\log \us}$ bits per key and query time $\Oh{\log \us}$.

\subparagraph*{ZFast Trie.}
To construct a \emph{ZFast Trie} \cite{belazzougui2009monotone}, we first generate a path-compacted trie.
Then, for prefixes of a specific length (\emph{2-fattest number}) of all input keys, a dictionary stores the trie node that represents that prefix.
A query can then perform a binary search over the length of the queried key.
If there is no node in the dictionary for a given prefix, the search can continue with the pivot as its upper bound.
If there is a node, the lower bound of the search can be set to the length of the longest common prefix of all keys represented by that node.
The ZFast trie uses $\Oh{\log\log\log \us}$ bits per key and query time $\Oh{(\log \us)/w + \log\log \us}$.

\subparagraph*{Path Decomposed Trie.}\label{s:pathDecomposed}
In the previous paragraphs, we described binary tries with a rather high height.
However, those tries are inefficient to query because of the pointer chasing to non-local memory areas.
The main idea behind \emph{Path Decomposed Tries} \cite{Ferragina:2008}, which can be used as an MMPHF \cite{grossi2014decomposition}, is to reduce the height of the tries.
We first select one path all the way from the root node to a leaf.
This path is now contracted to a single node, which becomes the root node in our new path decomposed trie.
The remaining nodes in the original trie form subtries branching from every node in that path.
We take all of these subtries, make them children of the root node, and annotate them by their branching character with respect to the selected path.
The subtries are then converted to path decomposed tries recursively.
In \emph{centroid} path decomposition, the path to be contracted is always the one that descends to the node with the most leaves in its subtree.

\section{\lemon}\label{s:lemon}
We now introduce the main contribution of this paper\,---\,the MMPHF \lemon.
The core idea of \lemon is surprisingly simple.
We take all the $n$ input integers and map them to $n$ buckets using some \emph{monotone} mapping function,  that we will describe later.
We store an Elias-Fano coded sequence with the \emph{global ranks} of the first key in each bucket using $2n+o(n)$ bits.
Given a bucket of size $b$, we use a $\lceil\log b\rceil$-bit retrieval data structure (see \cref{s:retrieval}) to store the \emph{local ranks} of all its keys. 
Note that we do not need to store local ranks if the bucket has only $0$ or $1$ keys. For squeezing space, instead of storing one retrieval data structure per bucket, we store a collection of retrieval data structures so that the $i$th one stores the local ranks of all keys mapped to buckets whose size $b$ is such that $i = \lceil\log b\rceil$.
An illustration of the overall data structure is given in \cref{fig:integerIllustration}.

\begin{figure}[t]
  \centering
  \begin{subfigure}[b]{0.49\textwidth}
    \centering
    \includegraphics[scale=0.8]{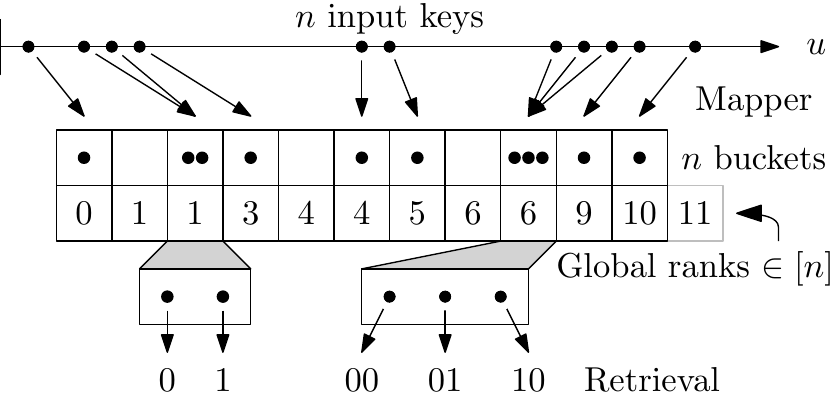}
    \caption{\lemon. Keys are mapped to buckets. Ranks within buckets are stored in (a collection of) retrieval data structures.}
    \label{fig:integerIllustration}
  \end{subfigure}
  \hfill
  \begin{subfigure}[b]{0.49\textwidth}
    \centering
    \includegraphics[scale=0.8]{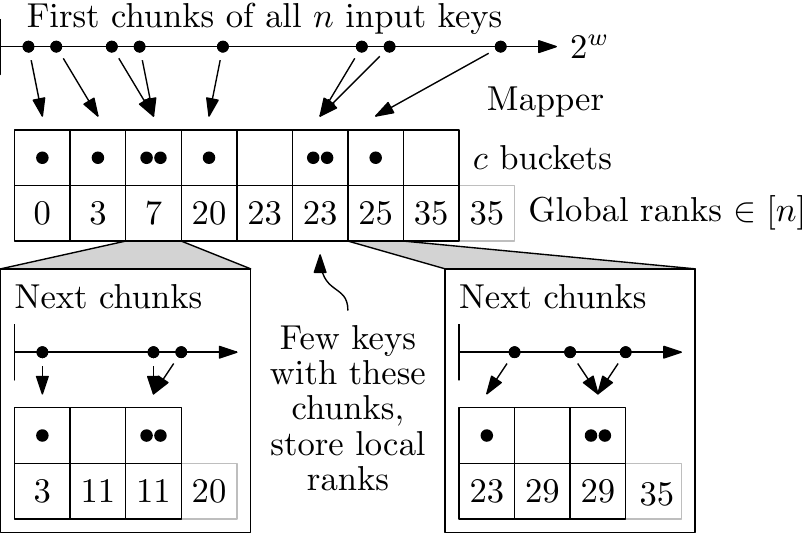}
    \caption{\lemonvl. Global ranks in each level are stored together. Buckets that are not handled recursively use retrieval data structures like before.}
    \label{fig:recursiveIllustration}
  \end{subfigure}
  \caption{Illustration of the \lemon and \lemonvl data structures.}
\end{figure}

\subparagraph*{Bucket Mapping Function.}\label{s:bucketMappingFunction}
The space efficiency of \lemon is directly related to the quality of the monotone mapping function.
For uniform random integers, a linear mapping from input keys to $n$ buckets, i.e. a mapping from a key $x$ to the bucket number $\lfloor xn/\us \rfloor$, leads to an MMPHF with a space usage of just $2.915$ bits per key (see \cref{thm:linearMappingSpace}).
Intuitively, such a linear mapping returns a rank estimate in $[n]$ for a given key.
However, for skewed distributions, the rank estimate can be far away which can create large buckets whose local ranks are expensive to store. 
For example, if the majority of the keys are such that $x<\us/n$, then the first bucket will be large enough to require $\Theta(\log n)$~bits per key, i.e. our MMPHF degenerates to a trivial OPMPHF.
To tackle this problem, we implement the mapping function with a \emph{PGM-index} \cite{ferragina2020pgm}.
As we observed in \cref{s:pgm}, the PGM was originally designed as a predecessor-search data structure. Here, we use the PGM as a {\em rank estimator} that, for a given key, returns an $\eps$-bounded estimate of its rank. To achieve this result in \lemon, we do not store the list of indexed keys and simply use the PGM's rank estimate as the bucket index.
The PGM internally adapts to the input data by learning the smoothness in the distribution via a piecewise linear $\eps$-approximation model, thus it can be thought of as a ``local'' approximation of the linear mapping above.
Real-world data sets can often be approximated using piecewise linear models, as discussed in the literature \cite{Ferragina:2021tcs} and also demonstrated by the good space efficiency of our experiments (see \cref{s:experiments}).
There is a trade-off between the amount of space needed to represent the PGM and the quality of the mapping, which depends on both the input data distribution and the given integer parameter~$\eps$.
In \cref{s:experiments}, we test both a version with a constant $\eps$ value and a version that auto-tunes its value by constructing multiple PGMs and then selecting the optimal~$\eps$.
Finally, we observe that with the PGM mapper, unlike for the linear mapping and other non error-bounded learning-based approaches~\cite{Kornaropoulos:2022,Ferragina:2020book}, the number of retrieval data structures we need to keep is bounded by $\Oh{\log \varepsilon}$ regardless of the input key distribution (see \Cref{thm:lemonPgm}).

\subparagraph*{Queries.}
\looseness=-1
Given a key $q$, we obtain its bucket $i$ using the mapping function.
The global rank of the (first key in the) bucket is the $i$th integer in the Elias-Fano coded sequence of global ranks, which can be accessed in constant time, and the bucket size is computed by subtraction from the next integer in that sequence.
The bucket size $b$ directly tells us which retrieval data structure to query, i.e. the $\lceil \log b \rceil$th one.
Evaluating the retrieval data structure with $q$ gives us its local rank in the bucket.
Adding this to the global rank of the bucket gives us the rank of $q$.
As we show in \cref{s:analysis}, for uniform data, the linear bucket mapper gives constant time queries, while for other inputs we use the PGM mapper and the query time is $\Oh{\log\log u}$.

\subparagraph*{Comparison to Known Solutions.}
Known MMPHFs in the literature typically divide the keys into equal-size buckets and build a compact trie-based distributor.
Unlike them, \lemon learns the data linearities and leverages them to distribute keys to buckets close to their rank.
Whenever some keys collide into a bucket, \lemon handles these keys via a (small) collection of succinct retrieval structures.
In contrast to known solutions, whenever a key is the only one mapped to its bucket, no information needs to be stored in (and no query is issued on) a retrieval data structure.
These features allow \lemon to possibly achieve reduced space occupancy compared to classic MMPHFs, which are oblivious to data linearities.
Also, \lemon can reduce the query time by replacing the cache-inefficient traversal of a trie with the PGM mapper, which in practice is fast to evaluate.

\section{\lemonvl}\label{s:lemonvl}
Of course, the idea of \lemon can be immediately applied to keys whose maximum longest common prefix (LCP) is less than $w$~bits. In this case, each string prefix and the following bit (which are sufficient to distinguish every string from each other) fit into one machine word and thus can be handled efficiently in time and in space by the PGM mapper.
For strings with longer LCPs, we introduce a tree data structure that we call \emph{\lemonvl} (since it handles Variable-Length strings).
The main idea is to simply compute the bucket mapping on a length-$w$ substring of each string, which we call a \emph{chunk}.
Buckets that receive many keys using this procedure are then handled recursively.
Details follow.

\subparagraph*{Overview.}

We start with a root node representing all the string keys in $S$ and consider the set of chunks extracted from each key starting from position $|p|$ (which we store), where $p$ is the LCP among the keys in $S$.
Given these $c$ distinct chunks, we construct a PGM mapper to distribute the keys to buckets in $[c]$, and we store an Elias-Fano coded sequence with the global ranks of the first key in each bucket.
Clearly, different keys can be mapped to the same bucket because the PGM mapper is not perfect (as in the integer case) and because they share the same chunk value (unlike in the integer case).
For example, for the strings $S=\{\texttt{cherry}, \texttt{cocoa}, \texttt{coconut}\}$ with $p=\texttt{c}$ and chunks composed of 3 characters, the keys  \texttt{cocoa} and \texttt{coconut} share the chunk value \texttt{oco} and will be mapped to the same bucket.

If a bucket of size $b$ contains fewer input strings than a specific threshold $t$, we store the local ranks of the strings in the bucket in a $\lceil \log b \rceil$-bit retrieval data structure.
Once again, we do not need to store local ranks if the bucket has only 0 or 1 keys.
If instead the bucket is large (i.e. $b \geq t$), we create a \emph{child node} in the tree data structure by applying the same idea recursively on the strings $S'$ of that bucket. This means that we compute a PGM mapper on the chunks extracted from each string in $S'$ starting from position $|p'|$, where $p'$ is the LCP among the bucket strings $S'$. 
Notice that $|p'|\geq |p|$ but we always guarantee that $S' \subsetneq S$, so the recursion is bounded.
In practice, we set the threshold $t=128$ (see \Cref{s:internalComparison}).

At query time, we can use the sequence of global ranks to calculate the bucket size $b$, which allows determining whether we need to continue recursively on a child (because $b\geq t$) or directly return the global rank of the bucket plus the local rank stored in the  $\lceil \log b \rceil$-bit retrieval data structure.
\Cref{fig:recursiveIllustration} gives an overview of the data structure.

We observe that the global ranks of each node increase monotonically from left to right in each level of the overall tree.
Therefore, we merge all these global ranks in a level into one Elias-Fano sequence, thereby avoiding the space overhead of storing many small sequences.

Of course, each inner node of the tree needs some extra metadata, like the encoding of its bucket mapper, the value of $|p|$, and an offset to its first global rank in the per-level Elias-Fano sequence.
We associate a node to its metadata via a minimal perfect hash function, where the identifier of a node is given by the path of the buckets' indices leading to it.

\label{s:optimisations}
Given the overall idea, there is a wide range of optimisations that we use. %
In the following, we outline the main algorithmic ones and refer the interested reader to our implementation~\cite{sourceCode} and \Cref{s:lowLevelOptimization} for the many other small-and-tricky optimisations, such as the use of specialised instructions like \texttt{popcount} and \texttt{bextr}, or lookup tables.

\subparagraph*{Alphabet Reduction.}%
The number of nodes and the depth of \lemonvl depend on both the length and distribution of the input strings, and on how well the PGM mapper at each node can map strings to distinct buckets given their $w$-bit chunks.
Therefore, we should aim to fit as much information as possible in the $w$-bit chunks.
We do so by exploiting the fact that, in real-world data sets, often only a very small alphabet $\Sigma$ of branching characters distinguish the strings in each bucket, and that we do not care about the other characters.
We extract chunks from the suffix of each string starting from the position following the LCP $p$, as before, but interpret the suffix as a number in radix $\sigma=|\Sigma|$ where each character is replaced by its 0-based index in $\Sigma$ if present, or by 0 if not present.
For example, for a node on the strings $\{\texttt{shoppers}, \texttt{shopping}, \texttt{shops}\}$ whose LCP is $p = \texttt{shop}$, we would store the alphabet $\Sigma=\{\texttt{e}, \texttt{i}, \texttt{p}, \texttt{s}\}$ and map the suffix ``$\texttt{pers}$'' of ``$\texttt{shoppers}$'' to $\text{index}(\texttt{p})\sigma^3 + \text{index}(\texttt{e})\sigma^2 + \text{index}(\texttt{r})\sigma^1 + \text{index}(\texttt{s})\sigma^0 = 2\sigma^3 + 0\sigma^2 + 0\sigma^1 + 3\sigma^0$.
Observe that the chunks computed in this way still preserve the lexicographic order of the strings.
The number of characters we extract is computed to fit as many characters as possible in a $w$-bit word, i.e. $\lfloor w/\log \sigma \rfloor$ characters.
In our implementation over bytes, we store $\Sigma$ via a bitmap of size 128 or 256, depending on whether its characters are a subset of ASCII or not. %
Finally, we mention that a mapping from strings to numbers in radix~$\sigma$ has also been used to build compressed string dictionaries~\cite{Boffa:2022spire}, but the twist here is that we are considering only the alphabet of the branching characters since we do not need to store the keys.

\subparagraph*{Elias-Fano Sequences.}
The large per-level Elias-Fano sequences of global ranks have a very irregular structure.
For example, if many of the strings in a node share the same chunks, there is a large gap between two of the stored ranks.
We can deal with these irregularities and reduce the overall space usage by using partitioned Elias-Fano \cite{ottaviano2014partitioned}.
Furthermore, the PGM mappers do not always provide a very uniform mapping, which thus results in empty buckets.
An empty bucket corresponds to a duplicate offset value being stored in the Elias-Fano sequences (see e.g. the duplicate offset 23 in \cref{fig:recursiveIllustration}).
To optimise the space usage of such duplicates, we filter them out before constructing the partitioned Elias-Fano sequence.
We do this by grouping the stored numbers in groups of 3 numbers.
If all 3 numbers are duplicates of the number before that group, we do not need to store the group.
A bit vector with rank support indicates which groups were removed.

\subparagraph*{Perfect Chunk Mapping.}\label{s:perfectChunkMapping}
In many datasets, there might be only a small number of different chunks, even if the number of strings they represent is large.
For instance, chunks computed on the first bytes of a set of URLs might be a few due to the scarcity of hostnames, but each host may contain many distinct pages.
In these cases, instead of a PGM, it might be more space-efficient to build a (perfect) map from chunks to buckets in $[c]$ via a retrieval data structure taking $c\lceil \log c \rceil$ bits overall (plus a small overhead), where $c$ is the number of distinct chunks.
In practice, we apply this optimisation whenever $c < 128$ (see \cref{s:internalComparison}).

\subparagraph*{Comparison to Known Solutions.}
In essence, \lemonvl applies the idea of \lemon recursively to handle variable-length strings.
Therefore, unlike known solutions, it can leverage data linearities to distribute $w$-bit chunks from the input strings to buckets using small space, and use additional child nodes only whenever a bucket contains many strings that thus require inspecting the following chunks to be distinguished.
Additionally, it performs an adaptive alphabet reduction within the buckets to fit more information in the $w$-bit chunks, thus leveraging the presence of more regularities in the input data. 
Overall, these features result in a data structure that has a small height and is efficient to be traversed.

\section{Variants and Refinements}\label{s:variants}

\looseness=-1
\lemon can be refined in numerous ways, which we only mention briefly  due to space constraints.
Looking at a possible external memory implementation, \lemon can be constructed trivially by a linear sweep and queries are possible using a suitable representation of the predecessor and bucket-size data structures.
\lemon can also be constructed in parallel without affecting the queries, in contrast to the trivial parallelisation by partitioning the input.
In \lemonvl, extracting chunks from non-contiguous bytes reduces the height of the trees but has worse trade-offs in practice.
Finally, we present an alternative to storing the local ranks explicitly.
The idea is to recursively split the universe size of that bucket and record the number of keys smaller than that midpoint.
Despite its query overhead, this technique might be of general interest for MMPHFs.
Refer to \Cref{s:variantsAppendix} for details.

\section{Analysis}\label{s:analysis}
We now prove some properties of our \lemon data structure for integers.
In our analysis, we use succinct retrieval data structures taking $rn+o(n)$~bits per stored value and answering queries in constant time (see \Cref{s:retrieval} and~\cite{dillinger2022burr}).
Furthermore, since our bucket mappers need multiplications and divisions, we make the simplifying assumption $u=2^w$ to avoid dealing with the increased complexity of these arithmetic operations over large integers.

\begin{theorem}
  A \lemon data structure with a bucket mapper that simply performs a linear interpolation of the universe on a list of $n$ uniform random keys needs $\approx n(2.91536 + o(1))$ bits on average%
  \footnote{Numerically, we find that a better space usage of $\approx2.902n$ bits can be achieved by mapping the $n$ keys to only $\approx0.909n$ buckets, but this difference is irrelevant in practice. It is also interesting to note that this is close to the space requirements of most of the practical non-monotone MPHFs \cite{fox1992faster,belazzougui2009hash,pibiri2021pthash,muller2014retrieval,chapman2011meraculous,limasset2017fast,lehmann2022sichash,esposito2020recsplit,bez2022recsplit}. Using an MMPHF can be useful when indexing an array through an MPHF, because sorting the hash values can be more cache efficient than a large number of random accesses to the array.}
  and answers queries in constant time.
  \label{thm:linearMappingSpace}
\end{theorem}
\begin{proof}
  We approximate the number of keys per bucket using a Poisson distribution which results in $0.91536n + o(n)$ bits of space for the retrieval data structures.
  On top of that, an Elias-Fano coding of the global bucket ranks gives $2n + o(n)$ bits.
  Refer to \Cref{s:fullProofs} for the full proof.
\end{proof}

While this result is formally only valid for a global uniform distribution, for use in \lemon it suffices if each segment computed by the PGM-index is sufficiently smooth.
It need not even be uniformly random as long as each local bucket has a constant average size.
As long as the space for encoding the segments is in $\Oh{n}$ bits, we retain the linear space bound of \Cref{thm:linearMappingSpace}.
Moreover, the following worst-case analysis gives us a fallback position that holds regardless of any assumptions.

\begin{theorem}\label{thm:lemonPgm}
A \lemon data structure with the PGM mapper takes 
$n(\lceil\log(2\eps+1)\rceil + 2 + o(1)) + \Oh{m \log \tfrac{\us}{m}}$~bits of space in the worst case and answers queries in $\Oh{\log \log_w \tfrac{\us}{m}}$~time, where $m$ is the number of linear models in a PGM with an integer parameter $\eps\geq0$ constructed on the $n$ input keys.
\end{theorem}
\begin{proof}
  The basic idea is that the rank estimate returned by the PGM is guaranteed to be far from the correct rank by $\eps$, which limits the space of the retrieval data structures.
  The $\Oh{m \log \tfrac{\us}{m}}$-term in the space bound is given by a compressed encoding of the linear models in the PGM, and the query time is given by a predecessor search structure on the linear models' keys.
  Refer to \Cref{s:fullProofs} for the full proof.
\end{proof}

The worst-case bounds obtained in \cref{thm:lemonPgm} are hard to compare with the ones of classic MMPHF (see \cref{s:mmph}) due to the presence of $m$ (and $\eps$), which depends on (and must be tuned according to) the approximate linearity of the input data, which classic MMPHFs are oblivious to.%
\footnote{This happens also in other problems in which data is encoded with linear models~\cite{boffa2022learned,Ferragina:2022}.}
Refer to \cref{s:pgm} for bounds on $m$.
Our experiments show that we obtain better space or space close to the best classic MMPHFs, while being much faster (we use a weaker but practical predecessor search structure than the one in \cref{thm:lemonPgm}).
Refer to \cref{s:experiments} for details.

\section{Experiments}\label{s:experiments}
In the following section, we first compare different configurations of \lemon and \lemonvl before comparing them with competitors from the literature.

\subparagraph*{Experimental Setup.}\label{s:experimentalSetup}
We perform our experiments on an Intel Xeon E5-2670 v3 with a base clock speed of 2.3~GHz running Ubuntu 20.04 with Linux 5.10.0.
We use the GNU C++ compiler version 11.1.0 with optimisation flags \texttt{-O3 -march=native}.
As a retrieval data structure, we use BuRR \cite{dillinger2022burr} with $64$-bit ribbon width and 2-bit bumping info.
To store the bucket sizes, we use the select data structure by Kurpicz \cite{kurpicz2022pasta} in \lemon and Partitioned Elias-Fano \cite{ottaviano2014partitioned} in \lemonvl.
To map tree paths to the node metadata, we use the MPHF PTHash \cite{pibiri2021pthash}.
For the PGM implementation in \lemon, we use the encoding from~\cref{thm:lemonPgm} and use a predecessor search on the Elias-Fano sequence (\cref{s:eliasFano}).
In \lemonvl, since the number of linear models in a node is typically small, we encode them explicitly as fixed-width triples $(\textit{key}, \textit{slope}, \textit{intercept})$ and find the predecessor via a binary search on the keys.
All our experiments are executed on a single thread.
Because the variation is very small, we run each experiment only twice and report the average.
We run the Java competitors on OpenJDK 17.0.4 and perform one warm-up run for the just-in-time compiler that is not measured.
With this, the Java performance is expected to be close to C++ \cite{belazzougui2011theoryPractice}.
Because Java does not have an unsigned 64-bit integer type, we subtract $2^{63}$ from each input key to keep their relative order.

The code and scripts needed to reproduce our experiments are available on GitHub under the General Public License \cite{sourceCode,sourceCodeComparison}.

\begin{table}[t]
  \caption{Datasets used for the experiments, together with their length or average (\o) length. Top: real-world string datasets. Middle: real-world integer datasets. Bottom: synthetic integer datasets.}
    \label{tab:datasets}
    \setlength{\tabcolsep}{4.8pt}
      \begin{tabularx}{\columnwidth}{l r r X}
        \toprule
        Dataset     &  $n$ &        Length & Description \\ \midrule
        text        &  35M &  \o{} 11 bytes & Terms appearing in the text of web pages, GOV2 corpus~\cite{belazzougui2011theoryPractice} \\
        dna         & 367M &       32 bytes & 32-mer from a DNA sequence, Pizza\&Chili corpus~\cite{pizzachili} \\
        urls        & 106M & \o{} 105 bytes & Web URLs crawled from .uk domains in 2007 \cite{boldi2008urls} \\ \midrule
        5gram       & 145M &        32 bits & Positions of the most frequent letter in the BWT of a text file containing 5-grams found in books indexed by Google \cite{googleNgram,boffa2022learned} \\
        fb          & 200M &        64 bits & Facebook user IDs \cite{kipf2019sosd} \\
        osm         & 800M &        64 bits & OpenStreetMap locations \cite{kipf2019sosd} \\ \midrule
        uniform     & 100M &        64 bits & Uniform random \\
        normal      & 100M &        64 bits & Normal distribution ($\mu=10^{15}$, $\sigma^2=10^{10}$)\\
        exponential & 100M &        64 bits & Exponential distribution ($\lambda=1$, scaled with $10^{15}$) \\
        \bottomrule
      \end{tabularx}
\end{table}

\subparagraph*{Datasets.}
Our datasets, as in previous evaluations~\cite{belazzougui2011theoryPractice,grossi2014decomposition}, are a \emph{text} dataset that contains terms appearing in the text of web pages \cite{belazzougui2011theoryPractice} and \emph{urls} crawled from .uk domains in 2007 \cite{boldi2008urls}.
Additionally, we also test with \emph{dna} sequences consisting of 32-mers \cite{pizzachili}.
Regarding real-world integer datasets, \emph{5gram} contains positions of the most frequent letter in the BWT of a text file containing 5-grams found in books indexed by Google \cite{googleNgram,boffa2022learned}.
The \emph{fb} dataset contains Facebook user IDs \cite{kipf2019sosd} and \emph{osm} contains OpenStreetMap locations \cite{kipf2019sosd}.
As synthetic integer datasets, we use 64-bit \emph{uniform}, \emph{normal}, and \emph{exponential} distributions.
Refer to \cref{tab:datasets} for details.

\subsection{Tuning Parameters}\label{s:internalComparison}
In the following section, we compare several configuration parameters of \lemon and show how they provide a trade-off between space usage and performance.

\subparagraph*{\lemon.}
Different ways of mapping the keys to buckets have their own advantages and disadvantages.
\Cref{tab:mappers} gives measurements of the construction and query throughput, as well as the space consumption of different bucket mappers.
Our implementation of \lemon with a linear bucket mapper achieves a space usage of $2.94n$ bits, which is remarkably close to the theoretical space usage of $2.91n$ bits (see \cref{thm:linearMappingSpace}).
Of course, a global, linear mapping does not work for all datasets.
A bucket mapper that creates equal-width segments by interpolating between sampled keys (denoted as ``Segmented'' in the table) is fast to construct and query, and it achieves good space usage. But, as for the global linear mapping, this approach is not robust enough to manage arbitrary input distributions.
In particular, for this heuristic mapper, it is easy to come up with a worst-case input that degenerates the space usage.
Conversely, with the PGM mapper, \lemon still achieves $2.96n$ and $2.98n$ bits on uniform random integers but it is more performant and robust on other datasets (except on osm, where the heuristic mapper obtains a good enough mapping with only its equal-width segments, which are inexpensive to store).
In fact, we explicitly avoided heuristic design choices in our PGM mapper (such as sampling input keys, removing outliers, or using linear regression) to not inflate our performance on the tested datasets at the expense of robustness on unknown ones (see Ref. \cite{Kornaropoulos:2022}).
Finally, on most input distributions, auto-tuning the value of $\eps \in \{15, 31, 63\}$  does not have a large effect on the space usage.

\begin{table}[t]
  \caption{Comparison of different bucket mappers. The space usage is given in bits per key, the query throughput in kQueries/second, and the construction throughput (c.t.) in MKeys/second.}
    \label{tab:mappers}
    \centering
    
\setlength{\tabcolsep}{4.5pt}
\begin{tabular}{l r r r r r r r r r r r r}
    \toprule
    Dataset & \multicolumn{3}{c}{Linear mapper} & \multicolumn{3}{c}{PGM $\eps=\text{auto}$} & \multicolumn{3}{c}{PGM $\eps=31$} & \multicolumn{3}{c}{Segmented} \\
              \cmidrule(lr){2-4}                  \cmidrule(lr){5-7}                    \cmidrule(lr){8-10}                 \cmidrule(lr){11-13}
            & bpk & kq/s & c.t.                 & bpk & kq/s & c.t.                   & bpk & kq/s & c.t.                 & bpk & kq/s & c.t.             \\ \midrule
          5gram &  5.60 & 1833.5 & 6.2 & 2.62 & 1747.0 & 3.8 & 2.63 & 1779.4 & 8.5 & 2.64 & 2145.9 & 14.5 \\
             fb & 34.35 &    0.8 & 5.1 & 4.91 & 1156.1 & 2.8 & 4.91 & 1150.7 & 5.1 & 4.93 & 1441.3 &  7.2 \\
            osm & 12.92 & 1525.3 & 5.5 & 4.42 &  999.6 & 2.8 & 4.42 &  998.6 & 5.0 & 4.33 & 1272.9 &  6.8 \\ \midrule
        uniform &  2.94 & 3244.6 & 8.7 & 2.96 & 1903.3 & 3.5 & 2.98 & 1850.5 & 6.5 & 3.03 & 2192.0 &  8.7 \\
         normal & 34.27 &  105.3 & 4.8 & 2.95 & 1935.0 & 3.6 & 2.97 & 1858.0 & 6.6 & 3.00 & 1727.7 &  8.7 \\
    exponential &  5.42 & 2715.9 & 6.0 & 2.95 & 1876.9 & 3.6 & 2.98 & 1791.5 & 6.6 & 3.01 & 2085.1 &  8.8 \\
    \bottomrule
\end{tabular}

\end{table}

\begin{table}[t]
    \caption{Comparison of different variants of \lemonvl.
      The space usage is given in bits per key, the query throughput in kQueries/second, and the construction throughput (c.t.) in MKeys/second.
      Variants with and without alphabet reduction (AR), a special indexed variant (Idx, see \Cref{s:indexedChunk}), and a variant with fixed instead of auto-tuned parameter $\eps$ for the bucket mapper.
      }
    \label{tab:recursiveTable}
    \centering
    
\setlength{\tabcolsep}{4.5pt}
\begin{tabular}{l r r r r r r r r r r r r}
    \toprule
    Dataset & \multicolumn{3}{c}{$\eps=\text{auto}$, no AR} & \multicolumn{3}{c}{$\eps=\text{auto}$, AR} & \multicolumn{3}{c}{$\eps=63$, AR} & \multicolumn{3}{c}{Idx, $\eps=\text{auto}$, AR} \\
              \cmidrule(lr){2-4}                       \cmidrule(lr){5-7}                    \cmidrule(lr){8-10}                 \cmidrule(lr){11-13}
            & bpk & kq/s & c.t.                      & bpk & kq/s & c.t.                   & bpk & kq/s & c.t.                 & bpk & kq/s & c.t.                        \\ \midrule
    text & 6.52 & 1062.9 & 1.7 & 6.03 & 1005.8 & 1.6 & 6.08 & 1001.8 & 2.5 & 6.10 & 933.2 & 2.3 \\
     dna & 7.66 &  452.8 & 2.0 & 6.32 &  631.3 & 1.7 & 6.25 &  644.8 & 2.7 & 6.27 & 601.1 & 2.4 \\
    urls & 7.14 &  282.7 & 2.3 & 6.37 &  298.8 & 1.8 & 6.46 &  295.1 & 2.3 & 6.63 & 298.1 & 1.6 \\
    \bottomrule
\end{tabular}

\end{table}

\subparagraph*{\lemonvl.}
\Cref{tab:recursiveTable} lists the effect of alphabet reduction on the query and construction performance.
In general, alphabet reduction enables noticeable space improvements with only a small impact on the construction time.
For the dna dataset, which uses only 15 different characters, the alphabet reduction has the largest effect, saving 1.3 bits per key and simultaneously making the queries 40\% faster.
The faster queries can be explained by the reduced tree height.
Note that alphabet reduction makes the queries slightly slower for the other datasets.
The reason is that instead of one single \texttt{bswap} instruction for chunk extraction, it needs multiple arithmetic operations (including \texttt{popcount}) for each input character.
The indexed variant that builds chunks from the distinguishing bytes instead of a contiguous byte range (see \Cref{s:indexedChunk}) is slower to construct but does not show clear space savings, which can be explained by larger per-node metadata.
We also experimented with different thresholds for when to stop recursion, as well as the perfect chunk mapping (see \cref{s:perfectChunkMapping}).
Given that the space overhead from each bucket mapper is the same for all data sets, it is not surprising that the same threshold (128 keys) works well for all datasets (see \Cref{s:thresholdsAppendix}).
Finally, making the $\eps$ value of the PGM mapper constant instead of auto-tuned, we naturally get faster construction.
As in the integer case, one would expect a fixed $\eps$ value to always produce results that are the same or worse than the auto-tuned version.
This is not the case because, in the recursive setting, it is hard to estimate the effect of a mapper on the overall space usage.
Therefore, an $\eps$ value that needs more space locally can lead to a mapping that proves useful on a later level of the tree.
This is why $\eps=63$ can achieve better space usage than the auto-tuned version on the dna dataset.

\subsection{Comparison with Competitors}\label{s:comparison}

\begin{figure}[p]
    \centering
    \input{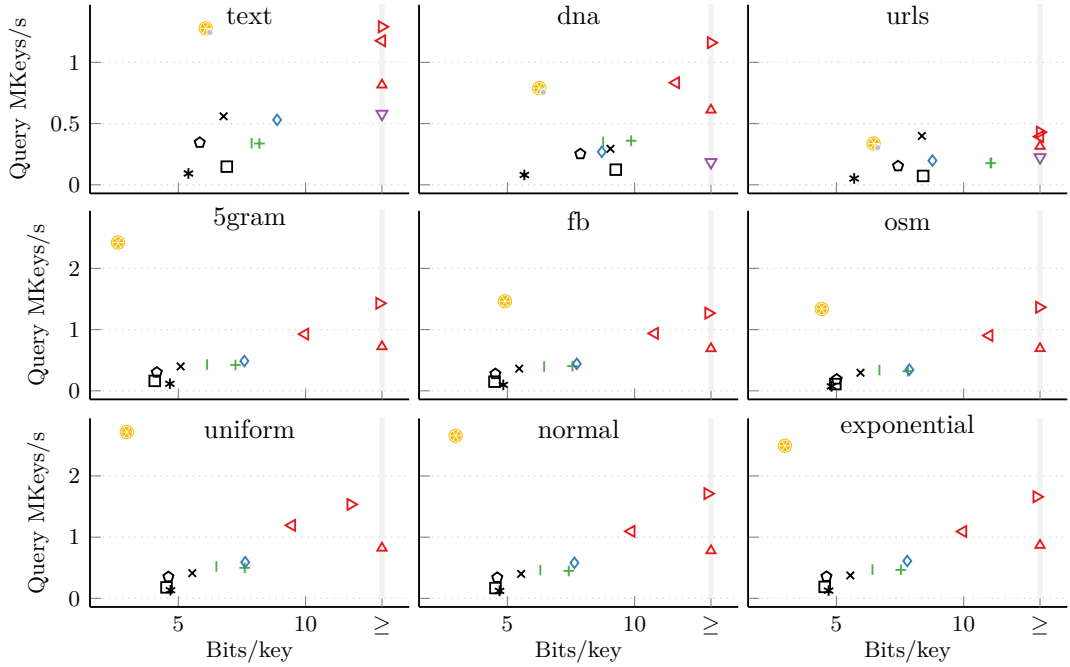}
    \caption{Query throughput for string, integer, and synthetic integer datasets vs space usage. The top-left corner of every plot shows the top-performing solutions in terms of space-time efficiency.}
    \label{fig:competitorsQuery}
\end{figure}

\begin{figure}[p]
    \centering
    \input{fig/construction.tex}

    \input{fig/competitorsLegend.tex}
    \caption{Construction throughput for string, integer, and synthetic integer datasets. Competitors with the \includegraphics[width=2mm]{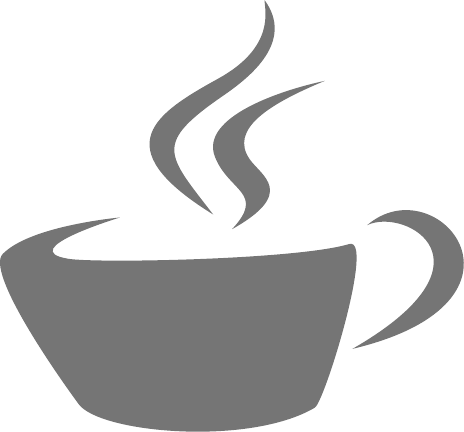} symbol in the legend are implemented in Java.}
    \label{fig:competitorsConstruction}
\end{figure}

In this section, we compare the performance of \lemon and \lemonvl with competitors from the literature.
Competitors include the C++ implementation by Grossi and Ottaviano~\cite{grossi2014decomposition} of the Centroid Hollow Trie, Hollow Trie, and Path Decomposed Trie.~%
Because that implementation only supports string inputs, we convert the integers to a list of fixed-length strings.
We point out that the Path Decomposed Trie crashes at an internal assertion when being run on integer datasets.
For the Hollow Trie, we encode the skips with either Gamma or Elias-Fano coding, whatever is better on the dataset. 
We also include the Java implementations by Belazzougui \etal~\cite{belazzougui2011theoryPractice} of a range of techniques (see \cref{s:bucketing}).
We use either the FixedLong or PrefixFreeUtf16 transformation, depending on the data type of the input.
For \lemon, we use the PGM mapper with $\eps=31$.
For \lemonvl, we use the PGM mapper with $\eps=63$, alphabet reduction and a recursion threshold $t=128$.

\subparagraph*{Queries.}
\Cref{fig:competitorsQuery} plots the query throughput against the achieved storage space.
In \Cref{tab:queryCompetitors} in the Appendix, we additionally detail the numbers in tabular format.
The LCP-based methods (see \cref{s:lcp}) have very fast queries but also need the most space (in fact, they appear to the top-right of the plots). At the same time, \lemon matches or even outperforms the query throughput of LCP-based methods, while being significantly more space efficient (in fact, it appears towards the top-left of the plots).
Compared to competitors with similar space usage, \lemon offers significantly higher query throughput.

\subparagraph*{Construction.}
\Cref{fig:competitorsConstruction} plots the construction throughput against the space needed.
On most synthetic integer datasets, \lemon provides a significant improvement to the state-of-the-art approaches, whereas it matches or outperforms the competitors on real-world datasets.
\lemon improves the construction throughput by up to a factor of \maxSpeedupConstructionNormalDistDRS{}, compared to the competitor with the next best space usage (typically, variants of the Hollow Trie).
While \lemonvl does not achieve the same space usage as the Hollow Trie Distributor, its construction is significantly faster, and still it is the second best in space usage.

\section{Conclusion and Future Work}\label{s:conclusion}
In this paper, we have introduced the monotone minimal perfect hash function \lemon.
\lemon, unlike previous solutions, learns and leverages data smoothness to obtain a small space usage and significantly faster queries.
On most synthetic and real-world datasets, \lemon dominates all competitors\,---\,simultaneously\,---\,on space usage, construction and query throughput.
Our extension to variable-length strings, \lemonvl, consists of trees that are significantly more flat and efficient to traverse than competitors.
This enables extremely fast queries with space consumption similar to competitors.

\subparagraph*{Future Work.}
Many MMPHF construction algorithms are based on the idea of explicitly storing ranks of keys within a small bucket.
The idea to split small buckets recursively that we mention in \Cref{s:variants} can help to reduce the space usage.
It remains an open problem whether the idea works in practice, especially when the distribution of keys inside the bucket is skewed.
It is also worth investigating a different construction of the piecewise linear approximation in the PGM that minimises the overall space given by the segments \emph{and} the local ranks stored in retrieval data structures, rather than the current approach that maximises the length of the segment (thus minimising just the segments space).
Applying non-linear transformations like low-degree polynomials within each segment would also be interesting future work.
Finally, it would be interesting to apply smoothed analysis to formally show that many real-world distributions locally behave as if they were uniform random, therefore leading to tighter space bounds.%

\bibliography{paper}

\begin{thebibliography}{10}

\bibitem{assadi2023tight}
Sepehr Assadi, Martin Farach{-}Colton, and William Kuszmaul.
\newblock Tight bounds for monotone minimal perfect hashing.
\newblock In {\em Proc. 34th Annual ACM-SIAM Symposium on Discrete Algorithms
  (SODA)}, pages 456--476, 2023.
\newblock \href {https://doi.org/10.1137/1.9781611977554.CH20}
  {\path{doi:10.1137/1.9781611977554.CH20}}.

\bibitem{belazzougui2009monotone}
Djamal Belazzougui, Paolo Boldi, Rasmus Pagh, and Sebastiano Vigna.
\newblock Monotone minimal perfect hashing: searching a sorted table with
  \emph{O}(1) accesses.
\newblock In {\em Proc. 20th Annual ACM-SIAM Symposium on Discrete Algorithms
  (SODA)}, pages 785--794, 2009.
\newblock \href {https://doi.org/10.1137/1.9781611973068.86}
  {\path{doi:10.1137/1.9781611973068.86}}.

\bibitem{belazzougui2011theoryPractice}
Djamal Belazzougui, Paolo Boldi, Rasmus Pagh, and Sebastiano Vigna.
\newblock Theory and practice of monotone minimal perfect hashing.
\newblock {\em {ACM} J. Exp. Algorithmics}, 16, 2011.
\newblock \href {https://doi.org/10.1145/1963190.2025378}
  {\path{doi:10.1145/1963190.2025378}}.

\bibitem{belazzougui2009hash}
Djamal Belazzougui, Fabiano~C. Botelho, and Martin Dietzfelbinger.
\newblock Hash, displace, and compress.
\newblock In {\em Proc. 17th Annual European Symposium on Algorithms (ESA)},
  pages 682--693, 2009.
\newblock \href {https://doi.org/10.1007/978-3-642-04128-0_61}
  {\path{doi:10.1007/978-3-642-04128-0_61}}.

\bibitem{belazzougui2020linear}
Djamal Belazzougui, Fabio Cunial, Juha K{\"{a}}rkk{\"{a}}inen, and Veli
  M{\"{a}}kinen.
\newblock Linear-time string indexing and analysis in small space.
\newblock {\em {ACM} Trans. Algorithms}, 16(2):17:1--17:54, 2020.
\newblock \href {https://doi.org/10.1145/3381417} {\path{doi:10.1145/3381417}}.

\bibitem{Belazzougui:2015}
Djamal Belazzougui and Gonzalo Navarro.
\newblock Optimal lower and upper bounds for representing sequences.
\newblock {\em {ACM} Trans. Algorithms}, 11(4):31:1--31:21, 2015.
\newblock \href {https://doi.org/10.1145/2629339} {\path{doi:10.1145/2629339}}.

\bibitem{BelazzouguiNV13}
Djamal Belazzougui, Gonzalo Navarro, and Daniel Valenzuela.
\newblock Improved compressed indexes for full-text document retrieval.
\newblock {\em J. Discrete Algorithms}, 18:3--13, 2013.
\newblock \href {https://doi.org/10.1016/j.jda.2012.07.005}
  {\path{doi:10.1016/j.jda.2012.07.005}}.

\bibitem{bez2022recsplit}
Dominik Bez, Florian Kurpicz, Hans{-}Peter Lehmann, and Peter Sanders.
\newblock High performance construction of {RecSplit} based minimal perfect
  hash functions.
\newblock In {\em Proc. 31st Annual European Symposium on Algorithms (ESA)},
  pages 19:1--19:16, 2023.
\newblock \href {https://doi.org/10.4230/LIPIcs.ESA.2023.19}
  {\path{doi:10.4230/LIPIcs.ESA.2023.19}}.

\bibitem{Boffa:2022spire}
Antonio Boffa, Paolo Ferragina, Francesco Tosoni, and Giorgio Vinciguerra.
\newblock Compressed string dictionaries via data-aware subtrie compaction.
\newblock In {\em Proc. 29th International Symposium on String Processing and
  Information Retrieval (SPIRE)}, pages 233--249, 2022.
\newblock \href {https://doi.org/10.1007/978-3-031-20643-6_17}
  {\path{doi:10.1007/978-3-031-20643-6_17}}.

\bibitem{boffa2022learned}
Antonio Boffa, Paolo Ferragina, and Giorgio Vinciguerra.
\newblock A learned approach to design compressed rank/select data structures.
\newblock {\em {ACM} Trans. Algorithms}, 18(3):24:1--24:28, 2022.
\newblock \href {https://doi.org/10.1145/3524060} {\path{doi:10.1145/3524060}}.

\bibitem{boldi2008urls}
Paolo Boldi, Massimo Santini, and Sebastiano Vigna.
\newblock A large time-aware web graph.
\newblock {\em {SIGIR} Forum}, 42(2):33--38, 2008.
\newblock \href {https://doi.org/10.1145/1480506.1480511}
  {\path{doi:10.1145/1480506.1480511}}.

\bibitem{boldyreva2011order}
Alexandra Boldyreva, Nathan Chenette, and Adam O'Neill.
\newblock Order-preserving encryption revisited: Improved security analysis and
  alternative solutions.
\newblock In {\em Proc. 31st Annual International Cryptology Conference
  (CRYPTO)}, pages 578--595, 2011.
\newblock \href {https://doi.org/10.1007/978-3-642-22792-9_33}
  {\path{doi:10.1007/978-3-642-22792-9_33}}.

\bibitem{chapman2011meraculous}
Jarrod~A. Chapman, Isaac Ho, Sirisha Sunkara, Shujun Luo, Gary~P. Schroth, and
  Daniel~S. Rokhsar.
\newblock Meraculous: De novo genome assembly with short paired-end reads.
\newblock {\em PLOS ONE}, 6(8):1--13, 08 2011.
\newblock \href {https://doi.org/10.1371/journal.pone.0023501}
  {\path{doi:10.1371/journal.pone.0023501}}.

\bibitem{clark1997compact}
David~Richard Clark.
\newblock {\em Compact Pat Trees}.
\newblock PhD thesis, University of Waterloo, Canada, 1996.

\bibitem{dillinger2022burr}
Peter~C. Dillinger, Lorenz H{\"{u}}bschle{-}Schneider, Peter Sanders, and
  Stefan Walzer.
\newblock Fast succinct retrieval and approximate membership using ribbon.
\newblock In {\em Proc. 20th International Symposium on Experimental Algorithms
  (SEA)}, pages 4:1--4:20, 2022.
\newblock \href {https://doi.org/10.4230/LIPICS.SEA.2022.4}
  {\path{doi:10.4230/LIPICS.SEA.2022.4}}.

\bibitem{dinklage2020practical}
Patrick Dinklage, Johannes Fischer, Alexander Herlez, Tomasz Kociumaka, and
  Florian Kurpicz.
\newblock Practical performance of space efficient data structures for longest
  common extensions.
\newblock In {\em Proc. 28th Annual European Symposium on Algorithms (ESA)},
  pages 39:1--39:20, 2020.
\newblock \href {https://doi.org/10.4230/LIPIcs.ESA.2020.39}
  {\path{doi:10.4230/LIPIcs.ESA.2020.39}}.

\bibitem{Elias74}
Peter Elias.
\newblock Efficient storage and retrieval by content and address of static
  files.
\newblock {\em J. {ACM}}, 21(2):246--260, 1974.
\newblock \href {https://doi.org/10.1145/321812.321820}
  {\path{doi:10.1145/321812.321820}}.

\bibitem{esposito2020recsplit}
Emmanuel Esposito, Thomas~Mueller Graf, and Sebastiano Vigna.
\newblock {RecSplit}: Minimal perfect hashing via recursive splitting.
\newblock In {\em Proc. 22nd Symposium on Algorithm Engineering and Experiments
  (ALENEX)}, pages 175--185, 2020.
\newblock \href {https://doi.org/10.1137/1.9781611976007.14}
  {\path{doi:10.1137/1.9781611976007.14}}.

\bibitem{Fano71}
Robert~Mario Fano.
\newblock On the number of bits required to implement an associative memory.
\newblock Technical report, MIT, Computer Structures Group, 1971.
\newblock Project MAC, Memorandum 61".

\bibitem{Ferragina:2008}
Paolo Ferragina, Roberto Grossi, Ankur Gupta, Rahul Shah, and Jeffrey~Scott
  Vitter.
\newblock On searching compressed string collections cache-obliviously.
\newblock In {\em Proc. 27th ACM Symposium on Principles of Database Systems
  (PODS)}, pages 181--190, 2008.
\newblock \href {https://doi.org/10.1145/1376916.1376943}
  {\path{doi:10.1145/1376916.1376943}}.

\bibitem{Ferragina:2021tcs}
Paolo Ferragina, Fabrizio Lillo, and Giorgio Vinciguerra.
\newblock On the performance of learned data structures.
\newblock {\em Theor. Comput. Sci.}, 871:107--120, 2021.
\newblock \href {https://doi.org/10.1016/J.TCS.2021.04.015}
  {\path{doi:10.1016/J.TCS.2021.04.015}}.

\bibitem{Ferragina:2022}
Paolo Ferragina, Giovanni Manzini, and Giorgio Vinciguerra.
\newblock Compressing and querying integer dictionaries under linearities and
  repetitions.
\newblock {\em {IEEE} Access}, 10:118831--118848, 2022.
\newblock \href {https://doi.org/10.1109/ACCESS.2022.3221520}
  {\path{doi:10.1109/ACCESS.2022.3221520}}.

\bibitem{pizzachili}
Paolo Ferragina and Gonzalo Navarro.
\newblock {Pizza\&Chili} corpus.
\newblock Accessed: February 2023.
\newblock URL: \url{http://pizzachili.dcc.uchile.cl/texts.html}.

\bibitem{Ferragina:2020book}
Paolo Ferragina and Giorgio Vinciguerra.
\newblock Learned data structures.
\newblock In Luca Oneto, Nicol{\`o} Navarin, Alessandro Sperduti, and Davide
  Anguita, editors, {\em Recent Trends in Learning From Data}, pages 5--41.
  Springer International Publishing, 2020.
\newblock \href {https://doi.org/10.1007/978-3-030-43883-8_2}
  {\path{doi:10.1007/978-3-030-43883-8_2}}.

\bibitem{ferragina2020pgm}
Paolo Ferragina and Giorgio Vinciguerra.
\newblock The {PGM-index}: a fully-dynamic compressed learned index with
  provable worst-case bounds.
\newblock {\em PVLDB}, 13(8):1162--1175, 2020.
\newblock \href {https://doi.org/10.14778/3389133.3389135}
  {\path{doi:10.14778/3389133.3389135}}.

\bibitem{Fox:1991}
Edward~A. Fox, Qi~Fan Chen, Amjad~M. Daoud, and Lenwood~S. Heath.
\newblock Order-preserving minimal perfect hash functions and information
  retrieval.
\newblock {\em {ACM} Trans. Inf. Syst.}, 9(3):281--308, 1991.
\newblock \href {https://doi.org/10.1145/125187.125200}
  {\path{doi:10.1145/125187.125200}}.

\bibitem{fox1992faster}
Edward~A. Fox, Qi~Fan Chen, and Lenwood~S. Heath.
\newblock A faster algorithm for constructing minimal perfect hash functions.
\newblock In {\em Proc. 15th Annual International ACM Conference on Research
  and Development in Information Retrieval (SIGIR)}, pages 266--273, 1992.
\newblock \href {https://doi.org/10.1145/133160.133209}
  {\path{doi:10.1145/133160.133209}}.

\bibitem{gagie2020fully}
Travis Gagie, Gonzalo Navarro, and Nicola Prezza.
\newblock Fully functional suffix trees and optimal text searching in
  {BWT}-runs bounded space.
\newblock {\em J. {ACM}}, 67(1):2:1--2:54, 2020.
\newblock \href {https://doi.org/10.1145/3375890} {\path{doi:10.1145/3375890}}.

\bibitem{sourceCode}
{LeMonHash - GitHub}.
\newblock \url{https://github.com/ByteHamster/LeMonHash}, 2023.

\bibitem{sourceCodeComparison}
{MMPHF-Experiments - GitHub}.
\newblock \url{https://github.com/ByteHamster/MMPHF-Experiments}, 2023.

\bibitem{gbmp2014sea}
Simon Gog, Timo Beller, Alistair Moffat, and Matthias Petri.
\newblock From theory to practice: Plug and play with succinct data structures.
\newblock In {\em Proc. 13th International Symposium on Experimental Algorithms
  (SEA)}, pages 326--337, 2014.
\newblock \href {https://doi.org/10.1007/978-3-319-07959-2_28}
  {\path{doi:10.1007/978-3-319-07959-2_28}}.

\bibitem{googleNgram}
Google.
\newblock Google ngram exports.
\newblock Accessed: March 2023.
\newblock URL:
  \url{https://storage.googleapis.com/books/ngrams/books/datasetsv3.html}.

\bibitem{grossi2010optimal}
Roberto Grossi, Alessio Orlandi, and Rajeev Raman.
\newblock Optimal trade-offs for succinct string indexes.
\newblock In {\em Proc. 37th International Colloquium on Automata, Languages
  and Programming (ICALP)}, pages 678--689, 2010.
\newblock \href {https://doi.org/10.1007/978-3-642-14165-2_57}
  {\path{doi:10.1007/978-3-642-14165-2_57}}.

\bibitem{grossi2014decomposition}
Roberto Grossi and Giuseppe Ottaviano.
\newblock Fast compressed tries through path decompositions.
\newblock {\em {ACM} J. Exp. Algorithmics}, 19(1), 2014.
\newblock \href {https://doi.org/10.1145/2656332} {\path{doi:10.1145/2656332}}.

\bibitem{Jacobson:1989}
Guy Jacobson.
\newblock Space-efficient static trees and graphs.
\newblock In {\em Proc. 30th IEEE Symposium on Foundations of Computer Science
  (FOCS)}, pages 549--554, 1989.
\newblock \href {https://doi.org/10.1109/SFCS.1989.63533}
  {\path{doi:10.1109/SFCS.1989.63533}}.

\bibitem{kipf2019sosd}
Andreas Kipf, Ryan Marcus, Alexander van Renen, Mihail Stoian, Alfons Kemper,
  Tim Kraska, and Thomas Neumann.
\newblock {SOSD:} {A} benchmark for learned indexes.
\newblock {\em CoRR}, abs/1911.13014, 2019.

\bibitem{Kornaropoulos:2022}
Evgenios~M. Kornaropoulos, Silei Ren, and Roberto Tamassia.
\newblock The price of tailoring the index to your data: Poisoning attacks on
  learned index structures.
\newblock In {\em Proc. 48th International Conference on Management of Data
  (SIGMOD)}, pages 1331--1344, 2022.
\newblock \href {https://doi.org/10.1145/3514221.3517867}
  {\path{doi:10.1145/3514221.3517867}}.

\bibitem{Kraska:2018}
Tim Kraska, Alex Beutel, Ed~H. Chi, Jeffrey Dean, and Neoklis Polyzotis.
\newblock The case for learned index structures.
\newblock In {\em Proc. 44th International Conference on Management of Data
  (SIGMOD)}, pages 489--504, 2018.
\newblock \href {https://doi.org/10.1145/3183713.3196909}
  {\path{doi:10.1145/3183713.3196909}}.

\bibitem{kurpicz2022pasta}
Florian Kurpicz.
\newblock Engineering compact data structures for rank and select queries on
  bit vectors.
\newblock In {\em Proc. 29th International Symposium on String Processing and
  Information Retrieval (SPIRE)}, pages 257--272, 2022.
\newblock \href {https://doi.org/10.1007/978-3-031-20643-6\_19}
  {\path{doi:10.1007/978-3-031-20643-6\_19}}.

\bibitem{kurpicz2023pachash}
Florian Kurpicz, Hans-Peter Lehmann, and Peter Sanders.
\newblock {PaCHash}: Packed and compressed hash tables.
\newblock In {\em Proc. 25th Symposium on Algorithm Engineering and Experiments
  (ALENEX)}, pages 162--175, 2023.
\newblock \href {https://doi.org/10.1137/1.9781611977561.ch14}
  {\path{doi:10.1137/1.9781611977561.ch14}}.

\bibitem{lehmann2022sichash}
Hans{-}Peter Lehmann, Peter Sanders, and Stefan Walzer.
\newblock {SicHash} - small irregular cuckoo tables for perfect hashing.
\newblock In {\em Proc. 25th Symposium on Algorithm Engineering and Experiments
  (ALENEX)}, pages 176--189, 2022.
\newblock \href {https://doi.org/10.1137/1.9781611977561.ch15}
  {\path{doi:10.1137/1.9781611977561.ch15}}.

\bibitem{lim2011silt}
Hyeontaek Lim, Bin Fan, David~G. Andersen, and Michael Kaminsky.
\newblock {SILT:} a memory-efficient, high-performance key-value store.
\newblock In {\em Proc. 23rd ACM Symposium on Operating Systems Principles
  (SOSP)}, pages 1--13, 2011.
\newblock \href {https://doi.org/10.1145/2043556.2043558}
  {\path{doi:10.1145/2043556.2043558}}.

\bibitem{limasset2017fast}
Antoine Limasset, Guillaume Rizk, Rayan Chikhi, and Pierre Peterlongo.
\newblock Fast and scalable minimal perfect hashing for massive key sets.
\newblock In {\em Proc. 16th International Symposium on Experimental Algorithms
  (SEA)}, pages 25:1--25:16, 2017.
\newblock \href {https://doi.org/10.4230/LIPICS.SEA.2017.25}
  {\path{doi:10.4230/LIPICS.SEA.2017.25}}.

\bibitem{majewski1996family}
Bohdan~S. Majewski, Nicholas~C. Wormald, George Havas, and Zbigniew~J. Czech.
\newblock A family of perfect hashing methods.
\newblock {\em Comput. J.}, 39(6):547--554, 1996.
\newblock \href {https://doi.org/10.1093/COMJNL/39.6.547}
  {\path{doi:10.1093/COMJNL/39.6.547}}.

\bibitem{muller2014retrieval}
Ingo M{\"{u}}ller, Peter Sanders, Robert Schulze, and Wei Zhou.
\newblock Retrieval and perfect hashing using fingerprinting.
\newblock In {\em Proc. 13th International Symposium on Experimental Algorithms
  (SEA)}, pages 138--149, 2014.
\newblock \href {https://doi.org/10.1007/978-3-319-07959-2_12}
  {\path{doi:10.1007/978-3-319-07959-2_12}}.

\bibitem{munro2001succinct}
J.~Ian Munro and Venkatesh Raman.
\newblock Succinct representation of balanced parentheses and static trees.
\newblock {\em {SIAM} J. Comput.}, 31(3):762--776, 2001.
\newblock \href {https://doi.org/10.1137/S0097539799364092}
  {\path{doi:10.1137/S0097539799364092}}.

\bibitem{navarro2014spaces}
Gonzalo Navarro.
\newblock Spaces, trees, and colors: The algorithmic landscape of document
  retrieval on sequences.
\newblock {\em {ACM} Comput. Surv.}, 46(4):1--47, 2014.
\newblock \href {https://doi.org/10.1145/2535933} {\path{doi:10.1145/2535933}}.

\bibitem{Navarro:2016book}
Gonzalo Navarro.
\newblock {\em Compact data structures: a practical approach}.
\newblock Cambridge University Press, 2016.

\bibitem{Navarro:2020pred}
Gonzalo Navarro and Javiel Rojas-Ledesma.
\newblock Predecessor search.
\newblock {\em {ACM} Comput. Surv.}, 53(5), 2020.
\newblock \href {https://doi.org/10.1145/3409371} {\path{doi:10.1145/3409371}}.

\bibitem{ottaviano2014partitioned}
Giuseppe Ottaviano and Rossano Venturini.
\newblock Partitioned {Elias-Fano} indexes.
\newblock In {\em Proc. 37th International ACM Conference on Research and
  Development in Information Retrieval (SIGIR)}, pages 273--282, 2014.
\newblock \href {https://doi.org/10.1145/2600428.2609615}
  {\path{doi:10.1145/2600428.2609615}}.

\bibitem{pibiri2021pthash}
Giulio~E. Pibiri and Roberto Trani.
\newblock {PTHash}: Revisiting {FCH} minimal perfect hashing.
\newblock In {\em Proc. 44th International ACM Conference on Research and
  Development in Information Retrieval (SIGIR)}, pages 1339--1348, 2021.
\newblock \href {https://doi.org/10.1145/3404835.3462849}
  {\path{doi:10.1145/3404835.3462849}}.

\bibitem{sabek2022can}
Ibrahim Sabek, Kapil Vaidya, Dominik Horn, Andreas Kipf, Michael Mitzenmacher,
  and Tim Kraska.
\newblock Can learned models replace hash functions?
\newblock {\em PVLDB}, 16(3):532--545, 2022.
\newblock \href {https://doi.org/10.14778/3570690.3570702}
  {\path{doi:10.14778/3570690.3570702}}.

\bibitem{vigna2008broadword}
Sebastiano Vigna.
\newblock Broadword implementation of rank/select queries.
\newblock In {\em Proc. 7th International Workshop on Experimental Algorithms
  (WEA)}, pages 154--168. Springer, 2008.
\newblock \href {https://doi.org/10.1007/978-3-540-68552-4_12}
  {\path{doi:10.1007/978-3-540-68552-4_12}}.

\bibitem{walzer2021peeling}
Stefan Walzer.
\newblock Peeling close to the orientability threshold - spatial coupling in
  hashing-based data structures.
\newblock In {\em Proc. 32nd ACM-SIAM Symposium on Discrete Algorithms (SODA)},
  pages 2194--2211, 2021.
\newblock \href {https://doi.org/10.1137/1.9781611976465.131}
  {\path{doi:10.1137/1.9781611976465.131}}.

\bibitem{Witten:1999}
Ian~H. Witten, Alistair Moffat, and Timothy~C. Bell.
\newblock {\em Managing Gigabytes: Compressing and Indexing Documents and
  Images}.
\newblock Morgan Kaufmann, 2nd edition, 1999.

\end{thebibliography}

\clearpage
\appendix
\renewcommand{\thefigure}{A.\arabic{figure}}
\setcounter{figure}{0}
\renewcommand{\thetable}{A.\arabic{table}}
\setcounter{table}{0}

\section{Variants and Refinements}\label{s:variantsAppendix}
The following section explains in detail our proposed variants and refinements inside the \lemon framework.

\subparagraph*{External Memory Construction.}
To construct the PGM-index with a specific $\eps$ value, a single scan over the input data is sufficient.
As soon as one of the segments is constructed, the corresponding keys can be mapped to buckets and the input for the retrieval data structures can be generated.
The retrieval data structures can be constructed in external memory as well \cite{dillinger2022burr}.
The construction of \lemon can therefore be performed entirely in external memory.
External memory queries are possible by selecting a suitable data structure for predecessor queries inside the PGM-index (such as the recursive structure in~\cite{ferragina2020pgm}), as well as an external-memory encoding of the bucket sizes.
\lemonvl can be constructed and queried in external memory using similar considerations.
While the recursion needs additional passes over the input data, note that the construction is performed in depth-first order, so it can profit from the locality between different levels.

\subparagraph*{Parallel Construction.}
As described in \cite{belazzougui2011theoryPractice}, it is easy to divide any MMPHF into multiple buckets (see \cref{s:bucketing}).
The buckets can then be constructed independently in parallel, but this naive construction introduces some query overhead due to adding another layer on top of the data structure.
Instead, the \lemon construction can be parallelised transparently to the queries.
We can divide the input data into ranges and construct independent PGM-indexes on each range.
When concatenating the linear models of all ranges, we get a PGM-index for the whole input set.
An advantage of this approach is that it is transparent to the queries.
With the naive division, this index stores a negligible number of additional segments linear in the number of processors, but these cut-points can likely be ``repaired'' locally, so that we do not get a space overhead for most inputs.
Mapping all keys to buckets by evaluating the PGM and therefore determining the input for the retrieval data structures is possible in parallel as well.
Finally, the retrieval data structures can be constructed in parallel.
This is again transparent to the queries and introduces only a negligible space overhead linear in the number of processors \cite{dillinger2022burr}.
For variable-length strings, each node of the \lemonvl construction can be parallelised just like described above.
On top of that, different child nodes can be constructed independently in parallel.

\subparagraph*{Recursive Bucket Splitting.}
Inside a bucket, our implementation explicitly stores the ranks of all keys.
Let us call this strategy \emph{Direct Rank Storing} (DRS).
An alternative method to determine the ranks within a bucket is \emph{Recursive Bucket Splitting} (RBS).
Take a bucket of size $b$ that can contain keys from the range $(L, R)$.
We can now split this bucket in half by storing how many of the keys are smaller than $M = (L+R)/2$.
This takes $\lceil \log_2(b+1)\rceil$ bits and splits the bucket into two sub-buckets of average size $b/2$.
The two sub-buckets can be handled recursively.
For uniform random inputs with an average bucket size of $b \geq 3$, RBS needs less space than DRS.
This reduction in space usage comes at the cost of more expensive query operations.
In particular, we need to query the retrieval data structures for every level in that bucket-internal tree.
An additional problem with this variant is that it depends on the distribution of keys.
In the worst case, when all key values are very close to $L$, the approach repeatedly needs to store the fact that $b$ keys are smaller than the midpoint.
This can lead to a space usage close to $\log(b)\log(R-L)$, which can be arbitrarily large depending on the universe size.
We therefore did not implement this construction for \lemon.%
Whether the RBS technique still works well with real-world data sets remains an open question.
Given that many MMPHF construction algorithms use the bucketing technique (see \cref{s:related}), the RBS technique might still be of general interest for MMPHFs.

\subparagraph*{Indexed Chunk Extraction.}\label{s:indexedChunk}
As described in \cref{s:lemonvl}, the chunks in \lemonvl are generated from consecutive characters.
Now consider an input where the positions of branching characters of the keys are very far.
Then the chunks encode a lot of data that is not necessary to differentiate the keys.
Instead, it is possible to determine the distinct minima of the LCP values of strings in the corresponding node.
Then chunks can be generated from the positions at these minima, which reduces the height of the tree.
In practice, however, we find that the plain version is faster and more space efficient (see \cref{s:internalComparison}).

\section{Full Proofs}
\label{s:fullProofs}

\begin{proof}[Full proof of \cref{thm:linearMappingSpace}]
 For $n$ uniform random integers mapped to $n$ buckets, the number of keys per bucket follows a binomial distribution with $p=1/n$.
 For large $n$, we can approximate this by the Poisson distribution with $\lambda=n\cdot1/n=1$.
 Therefore, the probability that a bucket has size $k$ is $\frac{\lambda^ke^{-\lambda}}{k!}=\frac{1}{k!\,e}$.
 Storing a bucket of size $k$ requires $k$ entries in the corresponding retrieval data structure, and each needs $\lceil\log k\rceil$ bits.
 Note that buckets of size $0$ and $1$ do not need to store ranks.
 Using the linearity of expectation, the average total number of bits to store in retrieval data structures is:

 $$\mathds{E}(\textrm{space}) = n\cdot\mathds{E}(\textrm{space per bucket}) = n\cdot\sum_{k=2}^\infty k\lceil\log k\rceil \cdot \frac{1}{k!\, e} \approx 0.91536n \textrm{.}$$

 A succinct retrieval data structure can then store this using $\approx 0.91536n + o(n)$ bits of space.
 The Elias-Fano coded sequence of global ranks takes $2n+o(n)$~bits.
 Overall, we get a space usage of $\approx n(2.91536+o(1))$ bits.

 For queries, the evaluation of the linear function and rounding can be executed in constant time.
 Now that we have the bucket index, we retrieve its offset and size from that binary sequence using two constant time $\textit{select}_1$ queries.
 From that, we know which retrieval data structure to query, and the actual query works in constant time \cite{dillinger2022burr}.
\end{proof}

\begin{proof}[Full proof of \cref{thm:lemonPgm}]
The rank estimate returned by the PGM is guaranteed to be far from the correct rank by $\eps$.
In other words, given a bucket number $i \in [n]$, any of the input keys with rank between $\max\{1,i-\eps\}$ and $\min\{i+\eps,n\}$ can be mapped to it, thus yielding a bucket of size at most $b=2\eps+1$.
In the worst case, there are $n/(2\varepsilon+1)$ of such size-$b$ buckets, which overall require storing $n$ local ranks in a $\lceil\log b\rceil$-bit retrieval data structure.
Additional $2n+o(n)$~bits are needed for the Elias-Fano coded sequence of global ranks.

The remaining term of the space bound is given by the PGM, that we encode with an Elias-Fano representation of linear models' $(x,y)$-endpoints in $m(\log\tfrac{\us}{m} + \log\tfrac{n}{m} + 2\log(2\eps+1)) + \Oh{m}$~bits~\cite{Ferragina:2022}.
 This can be bounded by $\Oh{m \log \tfrac{\us}{m}}$~bits, since from \cite[Lemma 2]{ferragina2020pgm} it holds $2\eps \leq n/m \leq u/m$.
Finally, we build the predecessor structure of~\cite[Theorem~A.1]{Belazzougui:2015} on the linear models' keys, which takes $\Oh{m\log\tfrac{\us}{m}}$~bits and yields a query time of $\Oh{\log \log_w \tfrac{\us}{m}}$.
\end{proof}

\section{Low-Level Optimizations}\label{s:lowLevelOptimization}
In addition to the main algorithmic optimizations described in the main part, we here detail some more low-level optimizations of our implementation.

We encode the alphabet reduction as a bitmap and use the \texttt{popcount} instruction to determine a character's index.
For determining how many characters fit into a chunk with a given alphabet, we use a lookup table of size 256 because that is more efficient than a (floating point) logarithm and division.
Depending on the dataset, multiple nodes of the tree might use alphabet reduction with a similar alphabet.
When constructing a node, we therefore look if another node stores a superset of the alphabet that still leads to the same number of characters fitting into a chunk, and possibly re-use the alphabet.
If no alphabet reduction is used, we use the \texttt{bswap} instruction to immediately convert the next 8 characters to a chunk.

To speed up access in Elias-Fano coded sequences, we use the \texttt{clz} instruction, which counts the number of leading zeroes in a word.
When calculating the LCP of strings, we do so for multiple bytes at once using 64-bit comparisons.
This general idea was already evaluated in Ref. \cite{dinklage2020practical}.
To avoid accessing the strings during alphabet map creation (which would lead to cache faults), we annotate the LCP array with the branching characters.

To decode the PGM metadata, which is stored as integers of small width, we use the \texttt{bextr} instruction to extract specific bits from a word.
To evaluate the PGM, we use a 64-bit division with overflow detection instead of a 128-bit division because in practice, 64 bits are often enough to store the operands.
For the PGM that auto-tunes its $\eps$ value, we abort early when we detect that the PGM itself is already larger than the optimal cost.
This way, very small $\eps$ values can often be ruled out earlier.

\begin{figure}[t]
   \centering
   \hfill
\begin{tikzpicture}[trim axis left]
    \begin{axis}[
        plotThresholds,title={text},
        legend to name=thresholdsLegend,
        legend columns=4,
        xmode=log,
      ]
      \addplot coordinates { (64,6.18946) (128,6.06034) (256,6.08174) (512,6.124) };
      \addlegendentry{32};
      \addplot coordinates { (64,6.20075) (128,6.0539) (256,6.07736) (512,6.12388) };
      \addlegendentry{64};
      \addplot coordinates { (64,6.58633) (128,6.077) (256,6.07444) (512,6.12264) };
      \addlegendentry{128};
      \addplot coordinates { (64,6.74658) (128,6.28859) (256,6.08125) (512,6.11958) };
      \addlegendentry{256};
    \end{axis}
\end{tikzpicture}
\begin{tikzpicture}[trim axis left]
    \begin{axis}[
        plotThresholds,title={dna},ylabel={},yticklabels={,,},
        xmode=log,
      ]
      \addplot coordinates { (64,6.34738) (128,6.23959) (256,6.2733) (512,6.29164) };
      \addlegendentry{32};
      \addplot coordinates { (64,6.36706) (128,6.23959) (256,6.27332) (512,6.29164) };
      \addlegendentry{64};
      \addplot coordinates { (64,6.7317) (128,6.2513) (256,6.27331) (512,6.29164) };
      \addlegendentry{128};
      \addplot coordinates { (64,6.88883) (128,6.45635) (256,6.27864) (512,6.29165) };
      \addlegendentry{256};

      \legend{};
    \end{axis}
\end{tikzpicture}
\begin{tikzpicture}[trim axis left]
    \begin{axis}[
        plotThresholds,title={urls},ylabel={},yticklabels={,,},
        xmode=log,
      ]
      \addplot coordinates { (64,6.72289) (128,6.54965) (256,6.69865) (512,7.05178) };
      \addlegendentry{32};
      \addplot coordinates { (64,6.63917) (128,6.45625) (256,6.61303) (512,6.97491) };
      \addlegendentry{64};
      \addplot coordinates { (64,7.02418) (128,6.46125) (256,6.56651) (512,6.91841) };
      \addlegendentry{128};
      \addplot coordinates { (64,7.23706) (128,6.72692) (256,6.5791) (512,6.8832) };
      \addlegendentry{256};

      \legend{};
    \end{axis}
\end{tikzpicture}

{\small Perfect Chunk Mapping threshold:}
\begin{tikzpicture}[baseline=0.6cm]
    \ref*{thresholdsLegend}
\end{tikzpicture}

   \caption{Different thresholds for when to store ranks (of keys and chunks) explicitly.}
   \label{fig:thresholds}
\end{figure}
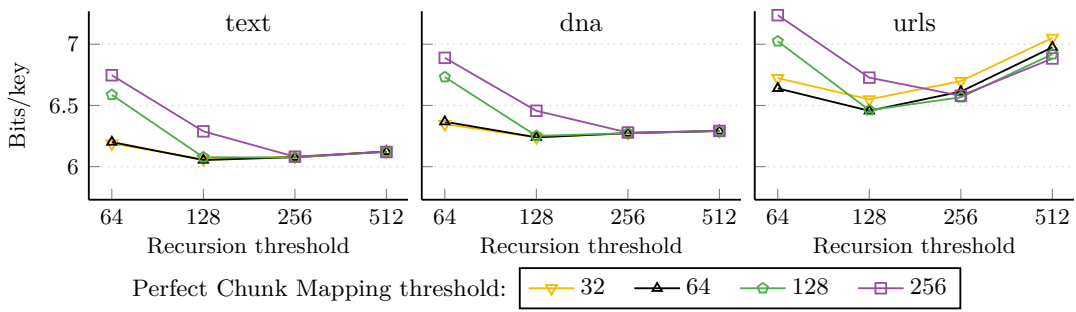

\section{Additional Experimental Data}
\Cref{tab:queryCompetitors} repeats the measurements of \cref{fig:competitorsConstruction,fig:competitorsQuery}, so that exact values can be compared.

\subparagraph*{Thresholds.}\label{s:thresholdsAppendix}
\Cref{fig:thresholds} compares different thresholds for when to stop recursion (see \cref{s:lemonvl}), as well as when to store ranks of chunks explicitly (see \cref{s:perfectChunkMapping}) in \lemonvl.
The behaviour of the different datasets is very similar, which is (as mentioned in \cref{s:internalComparison}) not surprising since the space overhead of a single-segment PGM bucket mapper is constant.
While we have not plotted the query performance here, note that queries get slightly faster when increasing the recursion threshold because that reduces the height of the tree.

\begin{sidewaystable}[t]
   \centering
   \caption{Comparison of string data sets. Query throughput is given in kQueries/s and space usage is given in bit/key (bpk).}
   \label{tab:queryCompetitors}
   
\setlength{\tabcolsep}{4.3pt}
\begin{tabular}{l r r r r r r r r r r r r r r r r r r}
    \toprule
    Method & \multicolumn{2}{c}{text} & \multicolumn{2}{c}{dna} & \multicolumn{2}{c}{urls} & \multicolumn{2}{c}{5gram} & \multicolumn{2}{c}{fb} & \multicolumn{2}{c}{osm} & \multicolumn{2}{c}{uniform} & \multicolumn{2}{c}{normal} & \multicolumn{2}{c}{exponential} \\
             \cmidrule(lr){2-3}         \cmidrule(lr){4-5}        \cmidrule(lr){6-7}         \cmidrule(lr){8-9}          \cmidrule(lr){10-11}     \cmidrule(lr){12-13}      \cmidrule(lr){14-15}          \cmidrule(lr){16-17}         \cmidrule(lr){18-19}
           & kq/s & bpk               & kq/s & bpk              & kq/s & bpk               & kq/s & bpk                & kq/s & bpk             & kq/s & bpk              & kq/s & bpk                  & kq/s & bpk                 & kq/s & bpk                      \\ \midrule
                                                 Centroid HT \cite{grossi2014decomposition} &                                                    560 &  6.78 &  294 &   9.05 & 399 &   8.36 &                                                 398 &  5.09 &  363 &  5.47 &  295 &  5.95 &  413 &  5.55 &  400 &  5.54 &  375 &  5.55 \\
        HTDist \includegraphics[width=2mm]{fig/coffee} \cite{belazzougui2011theoryPractice} &                                                     92 &  5.40 &   80 &   5.67 &  52 &   5.70 &                                                 115 &  4.67 &   97 &  4.84 &   73 &  4.81 &  133 &  4.69 &  122 &  4.69 &  127 &  4.69 \\
                                                      Hollow \cite{grossi2014decomposition} &                                                    345 &  5.84 &  252 &   7.87 & 153 &   7.42 &                                                 300 &  4.15 &  276 &  4.53 &  187 &  5.01 &  351 &  4.61 &  339 &  4.60 &  356 &  4.61 \\
        Hollow \includegraphics[width=2mm]{fig/coffee} \cite{belazzougui2011theoryPractice} &                                                    148 &  6.90 &  124 &   9.26 &  73 &   8.41 &                                                 162 &  4.07 &  150 &  4.50 &  110 &  4.96 &  179 &  4.54 &  169 &  4.53 &  188 &  4.54 \\
    LCP 2-step \includegraphics[width=2mm]{fig/coffee} \cite{belazzougui2011theoryPractice} &                                                   1176 & 13.12 &  834 &  11.62 & 394 &  17.81 &                                                 926 &  9.98 &  938 & 10.79 &  903 & 11.00 & 1193 &  9.46 & 1096 &  9.87 & 1093 &  9.97 \\
           LCP \includegraphics[width=2mm]{fig/coffee} \cite{belazzougui2011theoryPractice} &                                                   1291 & 21.61 & 1161 &  16.23 & 430 &  22.74 &                                                1429 & 12.90 & 1269 & 12.90 & 1364 & 12.97 & 1535 & 11.77 & 1711 & 12.87 & 1660 & 12.87 \\
          PaCo \includegraphics[width=2mm]{fig/coffee} \cite{belazzougui2011theoryPractice} &                                                    339 &  7.88 &  350 &   8.77 & 181 &  11.09 &                                                 429 &  6.13 &  397 &  6.44 &  340 &  6.69 &  522 &  6.50 &  463 &  6.30 &  471 &  6.42 \\
                                                Path Decomp. \cite{grossi2014decomposition} &                                                    579 & 54.44 &  185 & 148.27 & 224 & 228.88 & \makebox[1mm][l]{\hspace{3.8cm}crashes on integers} &       &      &       &      &       &      &       &      &       &      &       \\
         VLLCP \includegraphics[width=2mm]{fig/coffee} \cite{belazzougui2011theoryPractice} &                                                    816 & 18.43 &  611 &  20.13 & 315 &  22.59 &                                                 723 & 16.30 &  690 & 17.56 &  692 & 16.86 &  823 & 16.26 &  780 & 16.27 &  868 & 16.27 \\
        VLPaCo \includegraphics[width=2mm]{fig/coffee} \cite{belazzougui2011theoryPractice} &                                                    337 &  8.19 &  360 &   9.86 & 177 &  11.06 &                                                 423 &  7.25 &  404 &  7.56 &  320 &  7.81 &  500 &  7.61 &  449 &  7.41 &  465 &  7.53 \\
         ZFast \includegraphics[width=2mm]{fig/coffee} \cite{belazzougui2011theoryPractice} &                                                    530 &  8.88 &  269 &   8.71 & 198 &   8.77 &                                                 487 &  7.59 &  441 &  7.73 &  345 &  7.87 &  591 &  7.63 &  581 &  7.64 &  611 &  7.78 \\
                                                                      \textbf{LeMonHash-VL} &                                                   1278 &  6.08 &  790 &   6.25 & 338 &   6.46 &                                                1458 &  2.98 & 1111 &  4.91 &  857 &  4.39 & 1572 &  3.33 & 1647 &  3.32 & 1635 &  3.33 \\
                                                                         \textbf{LeMonHash} & \makebox[1mm][l]{\hspace{0.8cm}only supports integers} &       &      &        &     &        &                                                2421 &  2.63 & 1463 &  4.91 & 1338 &  4.42 & 2718 &  2.98 & 2657 &  2.97 & 2493 &  2.98 \\
    \bottomrule
\end{tabular}

\end{sidewaystable}

\end{document}